\documentclass[journal]{IEEEtran}

\usepackage{epsfig,amsmath,amssymb,epsf,amsthm,scalefnt,multirow,subfig}
\usepackage{xcolor}
\usepackage{float}
\usepackage{cite}
\usepackage{psfrag}
\usepackage{gensymb}
\usepackage{cleveref}
\usepackage{mathrsfs}
\usepackage{mathtools}
\usepackage{algpseudocode}
\usepackage{algorithm}
\usepackage{enumerate}
\usepackage{stfloats}
\usepackage{latexsym}
\usepackage{multicol}
\usepackage{pgfplots} 	
\usepackage{tikz}		

\newcommand{\argmax}{\mathop{\mathrm{argmax}}}

\newtheorem{lemma}{Lemma}
\newtheorem*{lemma*}{Lemma}

  \def\cC{{\mathcal{C}}}

\def\cM{{\mathcal{M}}} \def\cN{{\mathcal{N}}}

\def\argmax{\mathop{\mathrm{argmax}}}
\def\diag{\mathop{\mathrm{diag}}}

\def\b0{{\pmb{0}}} 

\def\ba{{\mathbf{a}}}   
 \def\bff{{\mathbf{f}}}  \def\bh{{\mathbf{h}}}
   
 \def\bn{{\mathbf{n}}}  \def\bp{{\mathbf{p}}}
\def\bq{{\mathbf{q}}}   
\def\bu{{\mathbf{u}}} \def\bv{{\mathbf{v}}} \def\bw{{\mathbf{w}}} \def\bx{{\mathbf{x}}}
\def\by{{\mathbf{y}}}   

\def\bA{{\mathbf{A}}} \def\bB{{\mathbf{B}}} \def\bC{{\mathbf{C}}} \def\bD{{\mathbf{D}}}
 \def\bF{{\mathbf{F}}}  \def\bH{{\mathbf{H}}}
\def\bI{{\mathbf{I}}}   
 \def\bN{{\mathbf{N}}}  
\def\bQ{{\mathbf{Q}}} \def\bR{{\mathbf{R}}}  
\def\bU{{\mathbf{U}}} \def\bV{{\mathbf{V}}} \def\bW{{\mathbf{W}}} 
\def\bY{{\mathbf{Y}}}

\def\red{\textcolor{black}}

\def\beginbreak{\begingroup\allowdisplaybreaks}

\begin{document}
	
\title{Practical Channel Estimation and Phase Shift Design for Intelligent Reflecting Surface Empowered MIMO Systems}

\author{Sucheol~Kim,
	Hyeongtaek~Lee,
	Jihoon~Cha,
	Sung-Jin~Kim,
	Jaeyong~Park,
	and~Junil~Choi
	
	\thanks{Sucheol~Kim, Hyeongtaek~Lee, Jihoon~Cha, and Junil~Choi are with the School of Electrical Engineering, Korea Advanced Institute of Science and Technology (e-mail: \{loehcusmik; htlee8459; charge; junil\}@kaist.ac.kr). Sucheol~Kim, Hyeongtaek~Lee, and Jihoon~Cha contributed equally to this work.}
	\thanks{Sung-Jin~Kim and Jaeyong~Park are with C\&M Standard Lab, Future Technology Center, LG Electronics Inc. (e-mail: \{sj88.kim; jaeyong630.park\}@lge.com).}
	
}


\maketitle

\begin{abstract}
	In this paper, channel estimation techniques and phase shift design for intelligent reflecting surface (IRS)-empowered single-user multiple-input multiple-output (SU-MIMO) systems are proposed. Among four channel estimation techniques developed in the paper, the two novel ones, single-path approximated channel (SPAC) and selective emphasis on rank-one matrices (SEROM), have low training overhead to enable practical IRS-empowered SU-MIMO systems. SPAC is mainly based on parameter estimation by approximating IRS-related channels as dominant single-path channels. SEROM exploits IRS phase shifts as well as training signals for channel estimation and easily adjusts its training overhead. A closed-form solution for IRS phase shift design is also developed to maximize spectral efficiency where the solution only requires basic linear operations. Numerical results show that SPAC and SEROM combined with the proposed IRS phase shift design achieve high spectral efficiency even with low training overhead compared to existing methods.
\end{abstract}

\section{Introduction} \label{Introduction}


\IEEEPARstart{I}{ntelligent} reflecting surface (IRS) is drawing great interest in recent years as a way to tackle the energy consumption problem of future wireless communication systems \cite{IRS_overview_1,IRS_overview_2,IRS_overview_3,IRS_overview_4,IRS_overview_5}. The IRS is a 2D surface consisting of low-cost passive scattering elements that can be deployed in an energy-efficient way and can present the benefit of array and beamforming gain as multiple antennas do in the multiple-input multiple-output (MIMO) systems. While there are no active elements in general, the IRS can manage magnitude (by turning on/off the passive elements) and phase shift of the incoming signals in order to strengthen the reflected signals achieving high spectral efficiency and overcoming large path-loss due to blockage \cite{IRS_overview_1},~\cite{IRS_overview_2}.

To fully exploit the advantages of the IRS-empowered communication systems, acquiring proper channel information on the IRS-related channels is essential at the base station (BS) or user equipment (UE). This is difficult in general since the IRS is not capable of transmitting or receiving training signals \cite{IRS_estimation_1,IRS_estimation_3}.
%
%
%
Several works were conducted to estimate the IRS-related channels. 
In \cite{T.L.Jensen:2020-MISO_CLRB}, an estimation technique was developed to minimize the Cram$\acute{\text{e}}$r-Rao lower bound of IRS-related channel estimation.
A least squares approach was adopted in \cite{D.Mishra:2019-MISO_estimation_power_transmit}, and the estimation error was analyzed with regard to the error offset on the IRS setting caused by the imperfect implementation. In \cite{C.You:2020-SISO_channel_estimation}, channel estimation under finite bit phase quantization of IRS elements was examined to afford a large number of IRS elements. These estimation techniques, though, are limited to multiple-input single-output or single-input single-output systems.

To realize high spectral efficiency that comes from spatial multiplexing, channel estimation for the IRS-empowered MIMO systems is necessary. In \cite{Z.He:020-SU_MIMO_two_stage}, the sparsity of MIMO channels was assumed, and the sparse matrix factorization and matrix completion were alternately repeated to construct estimated channels. The channel estimation in \cite{G.T.deAraujo:2020-LSKRF} utilized parallel factorization by reformulating the concatenation of received signals.
In \cite{Q.Nadeem:2020-MMSE-DFT}, minimum mean squared error (MMSE) estimation was developed based on the Rayleigh channel structure.
The above techniques, however, did not consider channel training overhead. 
%
The required training sequence length of the IRS-empowered communication systems could be much larger than the systems without the IRS due to a huge number of IRS elements \cite{Z.Wang:2020-low_overhead_estimation,IRS_estimation_2}.
%

Another important issue of using the IRS is to properly set the phases of IRS elements.
A min-rate maximization problem was formulated in \cite{Q.Nadeem:2020-MMSE-DFT}, and an iterative technique was proposed to obtain a sub-optimal solution. The IRS element design algorithm in \cite{IRS_overview_4} was developed to solve the proposed capacity characterization~problem. Though, the IRS element design in \cite{Q.Nadeem:2020-MMSE-DFT} is hard to be applied for data rate maximization, and the design in \cite{IRS_overview_4} does not guarantee its performance for imperfect channel information that is obtained by practical channel estimators.

In this paper, we configure a realistic IRS-empowered single-user MIMO (SU-MIMO) system. Considering the passive operation of IRS, we first express the cascaded channel through the IRS not in terms of two separate IRS-related channels, i.e., the UE-IRS link and IRS-BS link, but as a weighted sum of rank-one matrices. Based on the representation, we handle the two practical issues of this system: cascaded UE-IRS-BS channel estimation and IRS phase shift~design.

We develop four estimation techniques for the cascaded channel through the IRS. The first two techniques work as baselines while the last two techniques, single-path approximated channel (SPAC) and selective emphasis on rank-one matrices (SEROM), are novel with low training overhead. SPAC is developed by approximating the UE-IRS and IRS-BS links into dominant single-path channels. The BS estimates effective channel parameters to reconstruct the cascaded UE-IRS-BS channel, which largely reduces the training overhead compared to full channel matrix estimation. SEROM efficiently estimates the cascaded UE-IRS-BS channel by designing IRS reflection-coefficient matrices for training, which enables SEROM to easily adjusts its training overhead. Low-complexity IRS phase shift design with a closed-form solution is also proposed to maximize spectral efficiency.

We verify through simulations that SPAC and SEROM combined with the proposed IRS phase shift design achieve high spectral efficiency even with low training overhead. SPAC is specialized for the situation where the channel consists of a few dominant paths and is not affected from quantization of IRS phase shifts. SEROM achieves high spectral efficiency when the number of IRS elements is large because it can benefit from designing IRS reflection-coefficient matrices for training. The performance of proposed IRS phase shift design is comparable to that of exhaustive search with low computation complexity in the IRS-empowered SU-MIMO system. In terms of effective data transmissions, it is shown that the BS does not have to know the information of IRS-related channels separately to achieve high spectral efficiency.

The paper is organized as follows. In Section~\ref{System Model}, we explain the system model of SU-MIMO with the IRS. The estimation techniques for cascaded UE-IRS-BS channel are developed in Section~\ref{Channel Estimation Techniques}, and the low-complexity IRS phase shift design is proposed in Section~\ref{IRS Beamforming Techniques}. After presenting numerical results for channel estimation and IRS phase shift design in Section~\ref{numerical results}, we conclude the paper in Section~\ref{conclusion}.


\textit{Notations:} We use lower and upper boldface letters to represent column vectors and matrices. The element-wise conjugate, transpose, and conjugate transpose of a matrix $\bA$ are denoted by $\bA^\mathrm{*}$, $\bA^\mathrm{T}$, and $\bA^\mathrm{H}$, respectively. For a square matrix $\bA$, $\det(\bA)$, $\mathrm{Tr}(\bA)$, and $\bA^{-1}$ are the determinant, trace, and inverse of~$\bA$. $\bA(:,m:n)$ implies the submatrix that consists of the $m$-th column to the $n$-th column of the matrix~$\bA$, and the $m$-th element of a vector $\ba$ is denoted by $[\ba]_m$. $\angle(\ba)$ stands for the vector whose elements are phases of each element of a vector $\ba$. The diagonal matrix with the entries of a vector $\ba$ on its main diagonal is expressed as $\diag(\ba)$. The Kronecker product is denoted by $\otimes$, and $\odot$ implies the Hadamard product. $\boldsymbol{0}_m$ and $\boldsymbol{1}_m$ represent the $m \times 1$ all-zero vector and all-one vector, and $\bI_m$ represents the $m \times m$ identity matrix. $\cC\cN(\boldsymbol{\mu},\bQ)$ is used for the circularly symmetric complex Gaussian distribution with mean vector $\boldsymbol{\mu}$ and covariance matrix $\bQ$. Notations $\vert a \vert$ and $\mathrm{Re}(a)$ stand for the magnitude and real part of a complex number $a$. $\left\lceil a \right\rceil$ represents the minimum integer that is not smaller than a real number $a$. $\Vert\ba\Vert$ is the $\ell_2$-norm of a vector $\ba$, and $\Vert\bA\Vert_\mathrm{F}$ is the Frobenius-norm of a matrix $\bA$.


\begin{figure}[t]
	\centering
	\includegraphics[width=0.45\textwidth]{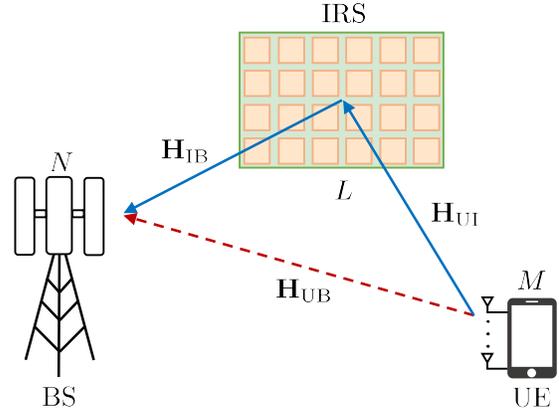}
	\caption{An IRS-empowered SU-MIMO communication system with $N$ BS antennas, $M$ UE antennas, and $L$ IRS elements.}
	\label{system model-fig}
\end{figure}

\section{System Model} \label{System Model} 
We consider an IRS-empowered time division duplexing (TDD) SU-MIMO system as shown in Fig. 1. The BS deploys~$N$ antennas and serves the UE equipped with $M$ antennas. The IRS, which consists of $L$ low-cost passive elements, is assumed to be connected to the BS via a controller where the BS is able to control the IRS elements for favorable signal reflection. For a practical setup, we consider a uniform planar array (UPA) for the BS and UE antennas and IRS elements.

During a channel coherence time block, the uplink received signal at the $t$-th time slot is \cite{Yang:2020_coherence}
\begin{align}
	\by_\mathrm{UL}[t]=\left(\bH_\mathrm{UB}+\bH_\mathrm{IB}\boldsymbol{\Phi}[t]\bH_\mathrm{UI}\right)\bff[t] s_\mathrm{UL}[t]+\bn_\mathrm{UL}[t],
\end{align}
where $s_\mathrm{UL}[t]\in \mathbb{C}$ is the transmit signal from the UE satisfying  $\mathbb{E}\{|s_\mathrm{UL}[t]|^2\}\leq P_\mathrm{UL}$ with the uplink transmit power $P_\mathrm{UL}$. The transmit beamformer $\bff[t] \in \mathbb{C}^{M \times 1}$ satisfies $\Vert\bff[t]\Vert^2=1$, and $\bn_\mathrm{UL}[t]\sim\cC\cN(\boldsymbol{0}_N,N_0\bI_N)$ is the thermal noise at the BS with the noise variance $N_0$. The uplink channels of the UE-BS direct link, IRS-BS link, and UE-IRS link are denoted by $\bH_{\mathrm{UB}}\in\mathbb{C}^{N\times M}$, $\bH_{\mathrm{IB}}\in\mathbb{C}^{N\times L}$, and $\bH_{\mathrm{UI}}\in\mathbb{C}^{L \times M}$, respectively. The $L \times L$ IRS reflection-coefficient matrix $\boldsymbol{\Phi}[t]$ is defined by $\diag\left(\left[\beta_1[t] e^{j\phi_1[t]},\cdots,\beta_L[t] e^{j\phi_L[t]}\right]^\mathrm{T}\right)$ where $\beta_\ell[t]$ and $\phi_\ell[t]$ are the magnitude and phase shift of the $\ell$-th IRS element. Considering practical passive operation of the IRS elements, we assume on/off magnitude $\beta_\ell[t]\in\{0,1\}$ and $B$-bit uniform quantization for each phase shift such that $\phi_\ell[t] \in \left\{ 0, \frac{2\pi}{2^B}, \cdots, \frac{(2^{B}-1)2\pi}{2^B}\right\}$. Since it is already shown in \cite{Phi_quantization} that $B \geq 4$ is enough to achieve almost the same performance of $B=\infty$, we first assume $B=\infty$ for conceptual explanation in Section~\ref{Channel Estimation Techniques} and Section~\ref{IRS Beamforming Techniques}. Then, the numerical results in Section~\ref{numerical results} are based on $B=4$ and $B=2$ for practically.

\begingroup\allowdisplaybreaks
We adopt the Rician fading with one line-of-sight (LoS) path and multiple non-line-of-sight (NLoS) paths for all channels \cite{Song:2017_Rician,Li:2019_Rician}. As an example, the uplink channel of the UE-BS direct link  $\bH_\mathrm{UB}$ is given~by
\begin{align} \label{channel model}
	&\bH_\mathrm{UB}=\sqrt{\mu_0\left(d_\mathrm{UB}/d_0\right)^{-\eta_{\mathrm{UB}}}} \sqrt{\frac{{NM}}{1+K_\mathrm{UB}}}	\notag \\	&\times
	\bigg(\sqrt{K_\mathrm{UB}} \alpha_{\mathrm{UB},0} \ba_\mathrm{BS}\left(\nu_{\mathrm{UB},0}^\mathrm{rx},\xi_{\mathrm{UB},0}^\mathrm{rx}\right)\ba_\mathrm{UE}^\mathrm{H}\left(\nu_{\mathrm{UB},0}^\mathrm{tx},\xi_{\mathrm{UB},0}^\mathrm{tx}\right)\notag \\ 
	&+\frac{1}{\sqrt{G_\mathrm{UB}}}\sum_{g=1}^{G_\mathrm{UB}}\alpha_{\mathrm{UB},g}\ba_\mathrm{BS}\left(\nu_{\mathrm{UB},g}^\mathrm{rx},\xi_{\mathrm{UB},g}^\mathrm{rx}\right)
	\notag \\&\times
	\ba_\mathrm{UE}^\mathrm{H}\left(\nu_{\mathrm{UB},g}^\mathrm{tx},\xi_{\mathrm{UB},g}^\mathrm{tx}\right)\bigg), 
\end{align}
where $\mu_0$ is the path-loss at the distance $d_0$, and the distance and path-loss exponent between the UE and BS are denoted by $d_\mathrm{UB}$ and $\eta_{\mathrm{UB}}$ \cite{IRS_estimation_2,PL_1,PL_2}. The Rician K-factor is denoted by $K_\mathrm{UB}$, and $G_\mathrm{UB}$ is the total number of NLoS paths. For the $g$-th path, $\alpha_{\mathrm{UB},g}\sim \cC\cN(0,1)$ is the complex path gain, and the vertical and horizontal arrival spatial frequencies at the BS are defined by $\nu_{\mathrm{UB},g}^\mathrm{rx} \triangleq \pi\sin(\theta_{\mathrm{UB},g}^\mathrm{rx})$ and $\xi_{\mathrm{UB},g}^\mathrm{rx} \triangleq \pi \sin(\psi_{\mathrm{UB},g}^\mathrm{rx})\cos(\theta_{\mathrm{UB},g}^\mathrm{rx})$ with the vertical and horizontal arrival angles $\theta_{\mathrm{UB},g}^\mathrm{rx}$ and $ \psi_{\mathrm{UB},g}^\mathrm{rx}$. Similarly, the vertical and horizontal departure spatial frequencies at the UE are defined by $\nu_{\mathrm{UB},g}^\mathrm{tx} \triangleq \pi\sin(\theta_{\mathrm{UB},g}^\mathrm{tx})$ and $\xi_{\mathrm{UB},g}^\mathrm{tx} \triangleq \pi  \sin(\psi_{\mathrm{UB},g}^\mathrm{tx})\cos(\theta_{\mathrm{UB},g}^\mathrm{tx})$ with the vertical and horizontal departure angles $\theta_{\mathrm{UB},g}^\mathrm{tx}$ and $ \psi_{\mathrm{UB},g}^\mathrm{tx}$. Assuming half wavelength spacing, the array response vectors at the BS and UE, i.e., $\ba_\mathrm{BS}(\cdot)$ and $\ba_\mathrm{UE}(\cdot)$, are given~as
\endgroup 
\begin{align}
	&\ba_\mathrm{BS}\left(\nu_{\mathrm{UB},g}^\mathrm{rx},\xi_{\mathrm{UB},g}^\mathrm{rx}\right) \notag\\
	=&\frac{1}{\sqrt{N}}\left[1,\cdots,e^{j(N_\mathrm{v}-1)\nu_{\mathrm{UB},g}^\mathrm{rx}} \right]^\mathrm{T} \otimes\left[ 1,\cdots,e^{j(N_\mathrm{h}-1)\xi_{\mathrm{UB},g}^\mathrm{rx}} \right]^\mathrm{T}, \label{BS array response vector} \\
	&\ba_\mathrm{UE}\left(\nu_{\mathrm{UB},g}^\mathrm{tx},\xi_{\mathrm{UB},g}^\mathrm{tx}\right) \notag \\ 
	=&\frac{1}{\sqrt{M}}\left[1,\cdots,e^{j(M_\mathrm{v}-1)\nu_{\mathrm{UB},g}^\mathrm{tx}} \right]^\mathrm{T} \otimes\left[ 1,\cdots,e^{j(M_\mathrm{h}-1)\xi_{\mathrm{UB},g}^\mathrm{tx}} \right]^\mathrm{T}, \label{UE array response vector}
\end{align}  
where $N=N_\mathrm{v}N_\mathrm{h}$ with $N_\mathrm{v}$ vertical and $N_\mathrm{h}$ horizontal antennas at the BS, and $M=M_\mathrm{v}M_\mathrm{h}$ with $M_\mathrm{v}$ vertical and $M_\mathrm{h}$ horizontal antennas at the UE. Note that $\bH_\mathrm{IB}$ and $\bH_\mathrm{UI}$ are modeled in the same way as in (\ref{channel model}) with proper adjustments on the distance, path-loss exponent, Rician K-factor, number of NLoS paths, number of antennas, array response vectors, and spatial frequencies.

Considering the reflection of incident signals at each IRS element, we can express the cascaded channel through the IRS as a weighted sum of rank-one matrices, which is given as
\begin{align}
	\bH_\mathrm{IB}\boldsymbol{\Phi}[t]\bH_\mathrm{UI}&=\sum_{\ell=1}^{L}\beta_\ell[t] e^{j\phi_\ell [t]} \bh_{\mathrm{IB},\ell}\bh_{\mathrm{UI},\ell}^\mathrm{H}, \label{sum of rank-1}
\end{align}
where $\bH_\mathrm{IB}=[\bh_{\mathrm{IB},1},\cdots,\bh_{\mathrm{IB},L}]$ and $\bH_\mathrm{UI}=[\bh_{\mathrm{UI},1},\cdots,\bh_{\mathrm{UI},L}]^\mathrm{H}$.
In (\ref{sum of rank-1}), the $\ell$-th rank-one matrix $\bh_{\mathrm{IB},\ell}\bh_{\mathrm{UI},\ell}^\mathrm{H}$ is weighted by $\beta_\ell[t] e^{j\phi_\ell [t]}$. For simplicity, we denote the $\ell$-th rank-one matrix as $\bR_{\ell}=\bh_{\mathrm{IB},\ell}\bh_{\mathrm{UI},\ell}^\mathrm{H}$, which gives
\begin{align} \label{one_rank_sum}
	\bH_\mathrm{IB}\boldsymbol{\Phi}[t]\bH_\mathrm{UI}=\sum_{\ell=1}^{L}\beta_\ell[t] e^{j\phi_\ell [t]}\bR_{\ell}.
\end{align}
The equality (\ref{one_rank_sum}) implies that it is sufficient to estimate the rank-one matrices $\bR_{\ell}$ instead of separately estimating $\bH_{\mathrm{IB}}$ and $\bH_{\mathrm{UI}}$. Hence, we will consider uplink channel estimation techniques for the direct channel $\bH_\mathrm{UB}$ and the rank-one matrices $\bR_{\ell}$ in Section~\ref{Channel Estimation Techniques}.


\section{IRS-Empowered MIMO Channel Estimation} \label{Channel Estimation Techniques}
In this section, we first explain the estimation of the direct link channel $ \bH_{\mathrm{UB}} $ and then elaborate on four channel estimation techniques to estimate the $ L $ rank-one matrices $ \bR_\ell $ in detail. Two rudimentary and straightforward techniques are described first as baselines, followed by two novel ones, SPAC and SEROM, which have low training overhead.

\subsection{UE-BS direct link channel estimation} \label{Direct Link Channel Estimation}
\beginbreak

To estimate the direct link channel $ \bH_{\mathrm{UB}}$, the BS turns off all the IRS elements as $ \boldsymbol{\Phi} [t]= \mathrm{diag}(\b0_L)$. The UE transmits the length $ \tau_{\mathrm{d}} $ training sequence using the training beamformer $ \bff[t] $ for $ 1 \leq t \leq \tau_{\mathrm{d}}$ with the training signal $ s_\mathrm{UL}[t]=\sqrt{P_\mathrm{UL}} $. By stacking the $ \tau_{\mathrm{d}} $ received signals,~we~have
\begin{align}
	\bY_{\mathrm{UB}}&=\left[\by_\mathrm{UL}[1],\cdots,\by_\mathrm{UL}[\tau_\mathrm{d}]\right]
	\notag \\&
	=\sqrt{P_\mathrm{UL}}\bH_{\mathrm{UB}}\bF_{\mathrm{UB}}+\bN_{\mathrm{UB}},
\end{align}
where $ \bF_{\mathrm{UB}}= \left[\bff[1],\cdots,\bff[\tau_{\mathrm{d}}]\right]$ is the training beamformer, and $ \bN_{\mathrm{UB}}= [\bn_\mathrm{UL}[1],\cdots,$ $\bn_\mathrm{UL}[\tau_{\mathrm{d}}]] $ is the noise. The training beamformer $ \bF_{\mathrm{UB}}$ can be composed of $ M $ rows of $ \tau_{\mathrm{d}}\times \tau_{\mathrm{d}} $ discrete Fourier transform (DFT) matrix with proper normalization, and we set $ \tau_{\mathrm{d}} =M $ to take the minimum sequence length such that $ \bF_{\mathrm{UB}}\bF_{\mathrm{UB}}^{\mathrm{H}}=\bI_M $. Then, the channel estimate for the direct link between the UE and BS is computed as
\begin{align}
	\widehat{\bH}_\mathrm{UB}
	&=\frac{1}{\sqrt{P_\mathrm{UL}}}\bY_\mathrm{UB} \bF_\mathrm{UB}^\mathrm{H}
	\notag\\
	&
	=\bH_\mathrm{UB} + \frac{1}{\sqrt{P_\mathrm{UL}}}\bN_\mathrm{UB}\bF_\mathrm{UB}^\mathrm{H}.
\end{align}

We define the additional training sequence length as $ \tau_{\mathrm{c}} $, which varies with estimation techniques, to estimate the cascaded UE-IRS-BS channel represented with $ \bR_\ell $. For $ \tau_{\mathrm{d}} + 1\leq t \leq \tau_{\mathrm{d}}+\tau_{\mathrm{c}} $, the BS eliminates the effect of direct link channel~as
\begin{align}
	\tilde{\by}_\mathrm{UL}[t]
	=&\by_\mathrm{UL}[t] - \sqrt{P_\mathrm{UL}}\widehat{\bH}_\mathrm{UB}\bff[t]\notag\\
	=&\sqrt{P_\mathrm{UL}}\bigg( \left(\bH_\mathrm{UB} - \widehat{\bH}_\mathrm{UB}\right)\notag\\ &+\bH_\mathrm{IB}\boldsymbol{\Phi}[t]\bH_\mathrm{UI}\bigg)\bff[t]+\bn_\mathrm{UL}[t]\notag\\
	=&\sqrt{P_\mathrm{UL}}\sum_{\ell=1}^{L}\beta_\ell[t]e^{j\phi_\ell [t]}\bR_{\ell}\bff[t]
	\notag\\	&
	+\underbrace{\sqrt{P_\mathrm{UL}}\left(-\bN_\mathrm{UB}\bF_\mathrm{UB}^\mathrm{H}\right)\bff[t]+\bn_\mathrm{UL}[t]}_{\triangleq \tilde{\bn}_\mathrm{UL}[t]}, 
	\label{Modified signal}
\end{align}
where $ \tilde{\bn}_\mathrm{UL}[t] $ is the effective uplink noise. We adopt the received signal $ \tilde{\by}_\mathrm{UL}[t] $ to explain the rank-one matrix estimation in the following subsections.

\subsection{One-by-one (OBO) channel estimation}
\label{Tech1}
The OBO estimation is to simply estimate $ L $ rank-one matrices one by one. The BS can estimate the $ \ell $-th rank-one matrix $\bR_\ell$ by turning on only the $ \ell $-th IRS element while keeping the others off and conduct this process in turn for each $ \ell $. Specifically, the IRS reflection-coefficient matrix $ \boldsymbol{\Phi}^{(\ell)} $ to estimate $\bR_\ell$ is defined as
\endgroup
\begin{align}\label{OBO IRS phase shift matrix}
	\boldsymbol{\Phi}^{(\ell)}=\diag\left(\left[\boldsymbol{0}_{\ell-1}^{\mathrm{T}},e^{j\phi_\ell},\boldsymbol{0}_{L-\ell}^{\mathrm{T}}\right]^{\mathrm{T}}\right),
\end{align}
and $ \boldsymbol{\Phi}[t] $ is fixed as $ \boldsymbol{\Phi}^{(\ell)} $ during the $ \ell $-th training period $ {\tau_{\mathrm{d}}+(\ell-1)M +1 \leq t \leq \tau_{\mathrm{d}}+\ell M} $. The UE transmits the length $ M $ training sequence with $ s_\mathrm{UL}[t]=\sqrt{P_\mathrm{UL}} $ using $ \bff[t] $ during the $ \ell $-th training period. The BS stacks the $ M $ received signals in \eqref{Modified signal} to estimate the $ \ell $-th rank-one matrix $\bR_\ell$, written~as
\begin{align}\label{Modified received signal for UIB_1}
	\widetilde{\bY}_{\mathrm{UIB},\ell} &= \left[\tilde{\by}_\mathrm{UL}[\tau_{\mathrm{d}}+(\ell-1)M+1],\cdots,\tilde{\by}_\mathrm{UL}[\tau_{\mathrm{d}}+\ell M]\right]
	\notag\\	&
	= \sqrt{P_\mathrm{UL}}e^{j\phi_\ell} \bR_\ell\bF_{\mathrm{UIB},\ell}+\widetilde{\bN}_{\mathrm{UIB},\ell},
\end{align}
where $ \bF_{\mathrm{UIB},\ell}\triangleq\left[\bff\left[\tau_{\mathrm{d}}+(\ell-1)M+1\right],\cdots,\bff\left[\tau_{\mathrm{d}}+\ell M\right]\right]$ and $ \widetilde{\bN}_{\mathrm{UIB},\ell}\triangleq[\tilde{\bn}_\mathrm{UL}[\tau_{\mathrm{d}}+(\ell-1)M+1],
\cdots,\tilde{\bn}_\mathrm{UL}\left[\tau_{\mathrm{d}}+\ell M\right]] $ are respectively the training beamformer and noise. As in Section~\ref{Direct Link Channel Estimation}, the normalized $ M\times M $ DFT matrix can be used as the training beamformer~$ \bF_{\mathrm{UIB},\ell} $.

With the $ M $ received signals in \eqref{Modified received signal for UIB_1}, the BS estimates the rank-one matrix for the $ \ell $-th IRS~element~as
\begin{align}
	\widehat{\bR}_\ell
	&=\frac{e^{-j\phi_\ell}}{\sqrt{P_\mathrm{UL}}}\widetilde{\bY}_\mathrm{UIB,\ell}  \bF_\mathrm{UIB,\ell}^\mathrm{H}
	\notag\\	&
	= \bR_\ell + \frac{e^{-j\phi_\ell}}{\sqrt{P_\mathrm{UL}}}\widetilde{\bN}_\mathrm{UIB,\ell}\bF_\mathrm{UIB,\ell}^\mathrm{H}.
\end{align}
Conducting this process for all $ L $ rank-one matrices,
the additional training sequence length for the OBO estimation becomes ${\tau_\mathrm{c}=\tau_{\mathrm{OBO}}=LM} $. It is obvious that the OBO estimation is inefficient since only one IRS element is turned on during each training period, resulting in significantly high training~overhead.

\subsection{Cooperative One-by-one (Co-OBO) channel estimation}
\label{Tech2}
With cooperative uplink and downlink signalings, the training sequence length can be made significantly small, compared to that of the OBO estimation in Section~\ref{Tech1}. The IRS reflection-coefficient matrix in the Co-OBO estimation is the same as in \eqref{OBO IRS phase shift matrix} but is employed only for two time slots for each $ \ell $. In other words, in order to estimate $ \bR_\ell $, $ \boldsymbol{\Phi}[t] $ is fixed as $ \boldsymbol{\Phi}^{(\ell)} $ during the $ \ell $-th training period $ \tau_{\mathrm{d}}+2(\ell-1)+1 \leq t \leq \tau_{\mathrm{d}}+2\ell $. The uplink and downlink signalings are sequentially conducted for the first and second time slots of each training period, i.e., the UE transmits $ s_\mathrm{UL}[\tau_{\mathrm{d}}+2(\ell-1)+1]=\sqrt{P_\mathrm{UL}} $, and the BS transmits $ s_\mathrm{DL}[\tau_{\mathrm{d}}+2\ell]=\sqrt{P_\mathrm{DL}} $ with the downlink transmit power $P_\mathrm{DL}$. For the first time slot, the $ \ell $-th uplink signal $ \tilde{\by}_{\mathrm{UL},\ell} \triangleq \tilde{\by}_\mathrm{UL}[ \tau_{\mathrm{d}}+2(\ell-1)+1] $ is expressed as
\begin{align}\label{Co-OBO UL}
	\tilde{\by}_{\mathrm{UL},\ell}
	&=
	e^{j\phi_\ell} \bh_{\mathrm{IB},\ell}\underbrace{\left(\sqrt{P_\mathrm{UL}}\bh_{\mathrm{UI},\ell}^\mathrm{H}\bff_{\mathrm{UL},\ell}\right)}_{\triangleq \tilde{s}_{\mathrm{UL},\ell}} + \tilde{\bn}_{\mathrm{UL},\ell},
\end{align}
where we define $ \bff_{\mathrm{UL},\ell} =\bff[\tau_{\mathrm{d}}+2(\ell-1)+1]$ and $ \tilde{\bn}_{\mathrm{UL},\ell} =\tilde{\bn}\left[\tau_{\mathrm{d}}+2(\ell-1)+1\right] $ for simplicity. The product $ \sqrt{P_\mathrm{UL}}\bh_{\mathrm{UI},\ell}^\mathrm{H}\bff_{\mathrm{UL},\ell} $ in \eqref{Co-OBO UL} can be regarded as the effective scalar-valued signal $ \tilde{s}_{\mathrm{UL},\ell} $.

Applying the channel reciprocity from TDD \cite{Larsson:2014_TDD}, the downlink received signal is
\begin{align}\label{DL received signal}
	\by_{\mathrm{DL}}[t]
	=
	\left(\bH_\mathrm{UB}^\mathrm{H}+\bH_\mathrm{IB}^\mathrm{H}\boldsymbol{\Phi}^\mathrm{H}[t]\bH_\mathrm{UI}^\mathrm{H}\right)\bw[t]s_\mathrm{DL}[t]+\bn_\mathrm{DL}[t],
\end{align}
where $ \bw[t] $ is the training beamformer at the BS satisfying $ \lVert \bw[t] \rVert^2 = 1$, and ${\bn_\mathrm{DL}[t]\sim\cC\cN(\boldsymbol{0}_M,N_0\bI_M)} $ is the thermal noise at the UE. Assuming perfect analog feedback, the UE feeds the received downlink signal $ \by_\mathrm{DL}[\tau_{\mathrm{d}}+2\ell] $ back to the BS. Then, similar to \eqref{Modified signal}, the BS can compute $ \tilde{\by}_{\mathrm{DL},\ell} $ as
\begin{align}\label{Co-OBO DL}
	&\tilde{\by}_{\mathrm{DL},\ell}\notag\\
	=& \by_\mathrm{DL}[\tau_{\mathrm{d}}+2\ell] - \sqrt{P_\mathrm{DL}}\widehat{\bH}_\mathrm{UB}^\mathrm{H}\bw_{\mathrm{DL},\ell}\notag\\
	=&\sqrt{P_\mathrm{DL}}\left( \left(\bH_\mathrm{UB}^\mathrm{H} - \widehat{\bH}_\mathrm{UB}^\mathrm{H}\right)+e^{-\phi_\ell}\bh_{\mathrm{UI},\ell}\bh_{\mathrm{IB},\ell}^\mathrm{H}\right)\bw_{\mathrm{DL},\ell}+\bn_{\mathrm{DL},\ell}\notag\\
	=&
	e^{-j\phi_\ell} \bh_{\mathrm{UI},\ell}\underbrace{\left(\sqrt{P_\mathrm{DL}}\bh_{\mathrm{IB},\ell}^\mathrm{H}\bw_{\mathrm{DL},\ell}\right)}_{\triangleq \tilde{s}_{\mathrm{DL},\ell}}
	\notag\\	&
	+\underbrace{\sqrt{P_\mathrm{DL}}\left(-\bN_\mathrm{UB}\bF_\mathrm{UB}^\mathrm{H}\right)^\mathrm{H}\bw_{\mathrm{DL},\ell}+\bn_{\mathrm{DL},\ell}}_{\triangleq \tilde{\bn}_{\mathrm{DL},\ell}}, 
\end{align}
where we define $ \bw_{\mathrm{DL},\ell} =\bw[\tau_{\mathrm{d}}+2\ell]$ and $ \bn_{\mathrm{DL},\ell} = \bn_\mathrm{DL}[\tau_{\mathrm{d}}+2\ell] $, and $ \tilde{\bn}_{\mathrm{DL},\ell} $ is the effective downlink noise. Again, the product $ \sqrt{P_\mathrm{DL}}\bh_\mathrm{IB,\ell}^\mathrm{H} \bw_{\mathrm{DL},\ell} $ in \eqref{Co-OBO DL} can be regarded as the effective scalar-valued signal~$ \tilde{s}_{\mathrm{DL},\ell} $.

Using $ \tilde{\by}_{\mathrm{UL},\ell} $ and $ \tilde{\by}_{\mathrm{DL},\ell} $ in \eqref{Co-OBO UL} and \eqref{Co-OBO DL}, the BS finally estimates the rank-one matrix as
\begin{align}
	\widehat{\bR}_\ell &= \frac{\tilde{\by}_{\mathrm{UL},\ell}\tilde{\by}_{\mathrm{DL},\ell}^\mathrm{H}}{(\sqrt{P_\mathrm{DL}}e^{-j\phi_\ell}\bw_{\mathrm{DL},\ell})^\mathrm{H}\tilde{\by}_{\mathrm{UL},\ell}}
	\notag\\	&
	=\frac{\bh_{\mathrm{IB},\ell}\tilde{s}_{\mathrm{UL},\ell}\tilde{s}^*_{\mathrm{DL},\ell}\bh_{\mathrm{UI},\ell}^\mathrm{H}}{\tilde{s}^*_{\mathrm{DL},\ell}\tilde{s}_{\mathrm{UL},\ell} + \sqrt{P_\mathrm{DL}}e^{-j\phi_\ell}\bw^\mathrm{H}_{\mathrm{DL},\ell}\tilde{\bn}_{\mathrm{UL},\ell}} +\widetilde{\bN}_\ell
	\notag\\	&
	=\frac{\bR_\ell}{1+\tilde{n}_\ell}+\widetilde{\bN}_\ell,
\end{align}
where $ \tilde{n}_\ell $ and $ \widetilde{\bN}_\ell $ are the noise terms, expressed as
\begin{align}
	\tilde{n}_\ell =& \frac{\sqrt{P_\mathrm{DL}}e^{-j\phi_\ell}\bw^\mathrm{H}_{\mathrm{DL},\ell}\tilde{\bn}_{\mathrm{UL},\ell}}{\tilde{s}_{\mathrm{UL},\ell}\tilde{s}^*_{\mathrm{DL},\ell}}, \\
	\widetilde{\bN}_\ell
	=&\frac{1}{1+\tilde{n}_\ell}
	\bigg(\frac{e^{-j\phi_\ell}\bh_{\mathrm{IB},\ell}\tilde{\bn}_{\mathrm{DL},\ell}^\mathrm{H}}{\tilde{s}^*_{\mathrm{DL},\ell}}
	\notag\\&
	+\frac{e^{-j\phi_\ell}\tilde{\bn}_{\mathrm{UL},\ell}\bh_{\mathrm{UI},\ell}^\mathrm{H}}{\tilde{s}_{\mathrm{UL},\ell}}+
	\frac{e^{-j2\phi_\ell}\tilde{\bn}_{\mathrm{UL},\ell}\tilde{\bn}_{\mathrm{DL},\ell}^\mathrm{H}}{\tilde{s}_{\mathrm{UL},\ell}\tilde{s}^*_{\mathrm{DL},\ell}}\bigg),
\end{align}
respectively. The Co-OBO estimation requires only two time slots to estimate $ {\bR}_\ell $ for each $ \ell $, which implies that the additional training sequence length for the Co-OBO estimation is $ \tau_\mathrm{c}=\tau_{\mathrm{Co-OBO}} = 2L $. For $ M \gg 2 $, which is valid for typical MIMO systems, it is obvious that $ \tau_{\mathrm{Co-OBO}} \ll \tau_\mathrm{OBO}=LM $. However, employing only two training signals to estimate each rank-one matrix makes the Co-OBO estimation vulnerable to burst noise, and perfect analog feedback is difficult to achieve in practice as well.

\subsection{Single-path approximated channel (SPAC)}
\label{SPAC}
We propose SPAC to overcome the high training overhead of OBO estimation and the burst noise issue of Co-OBO estimation. SPAC is developed to consider the structural property of IRS-empowered system and to extract the necessary channel parameters. SPAC estimates the rank-one matrices by approximating $ \bH_\mathrm{IB} $ and $ \bH_\mathrm{UI} $ as dominant single-path~channels.

The single-path approximations for $ \bH_\mathrm{IB} $ and $ \bH_\mathrm{UI} $ are expressed as
\begin{align}
	\bH_\mathrm{IB}&\approx \widetilde{\bH}_\mathrm{IB} =
	\gamma_{\mathrm{IB}}
	\ba_\mathrm{BS}\left(\nu_{\mathrm{IB}}^\mathrm{rx},\xi_{\mathrm{IB}}^\mathrm{rx}\right)\ba_\mathrm{IRS}^\mathrm{H}\left(\nu_{\mathrm{IB}}^\mathrm{tx},\xi_{\mathrm{IB}}^\mathrm{tx}\right), \label{SPAC_IB}\\
	\bH_\mathrm{UI}&\approx\widetilde{\bH}_\mathrm{UI} = 
	\gamma_{\mathrm{UI}}
	\ba_\mathrm{IRS}\left(\nu_{\mathrm{UI}}^\mathrm{rx},\xi_{\mathrm{UI}}^\mathrm{rx}\right)\ba_\mathrm{UE}^\mathrm{H}\left(\nu_{\mathrm{UI}}^\mathrm{tx},\xi_{\mathrm{UI}}^\mathrm{tx}\right). \label{SPAC_UI}
\end{align}
Focusing on \eqref{SPAC_IB}, $ \gamma_{\mathrm{IB}} $ is the effective complex-valued gain between the IRS and BS. Similar to \eqref{BS array response vector} and \eqref{UE array response vector}, the array response vector at the IRS is given as
\begin{align}
	\label{Approx_arv_IRS}
	&\ba_\mathrm{IRS}(\nu_{\mathrm{IB}}^\mathrm{tx},\xi_{\mathrm{IB}}^\mathrm{tx})
	\notag\\	=&
	\frac{1}{\sqrt{L}}\left[1,\cdots,e^{j(L_\mathrm{v}-1)\nu_{\mathrm{IB}}^\mathrm{tx}} \right]^\mathrm{T}\otimes\left[ 1,\cdots,e^{j(L_\mathrm{h}-1)\xi_{\mathrm{IB}}^\mathrm{tx}} \right]^\mathrm{T},
\end{align}
with the vertical and horizontal spatial frequencies $ \nu_{\mathrm{IB}}^\mathrm{tx}$ and $\xi_{\mathrm{IB}}^\mathrm{tx} $. The numbers of vertical and horizontal IRS elements are denoted by $ L_\mathrm{v} $ and $ L_\mathrm{h} $ satisfying $ L=  L_\mathrm{v}L_\mathrm{h}$. The parameters in \eqref{SPAC_UI} are similarly defined.

We can estimate the two gains and eight spatial frequencies embedded on the rank-one~matrix
\begin{align}
	\widetilde{\bR}_\ell =&\widetilde{\bH}_{\mathrm{IB}}(:,\ell)\widetilde{\bH}_{\mathrm{UI}}(\ell,:)
	\notag\\ =&
	\gamma_{\mathrm{IB}}\ba_\mathrm{BS}\left(\nu_{\mathrm{IB}}^\mathrm{rx},\xi_{\mathrm{IB}}^\mathrm{rx}\right)[\ba_\mathrm{IRS}^\mathrm{H}\left(\nu_{\mathrm{IB}}^\mathrm{tx},\xi_{\mathrm{IB}}^\mathrm{tx}\right)]_\ell  
	\notag \\	&\times 
	\gamma_{\mathrm{UI}}
	[\ba_\mathrm{IRS}\left(\nu_{\mathrm{UI}}^\mathrm{rx},\xi_{\mathrm{UI}}^\mathrm{rx}\right)]_\ell\ba_\mathrm{UE}^\mathrm{H}\left(\nu_{\mathrm{UI}}^\mathrm{tx},\xi_{\mathrm{UI}}^\mathrm{tx}\right),\label{SPAC_cascaded_prototype}
\end{align}
for each $ \ell $. The novel part of SPAC is that the BS does not estimate all the parameters in \eqref{SPAC_cascaded_prototype} separately but acquire the effective parameters concerned with them. The overall process of SPAC is summarized as follows: 
\begin{enumerate}[\indent Step  1:]
	\item By sequentially turning on only a small number of IRS elements one by one, a few rank-one matrices $ \bR_\ell $ are estimated by the OBO estimation.
	\item The spatial frequencies $ \left({\nu}_{\mathrm{IB}}^\mathrm{rx},{\xi}_{\mathrm{IB}}^\mathrm{rx}\right)$ and $\left({\nu}_{\mathrm{UI}}^\mathrm{tx},{\xi}_{\mathrm{UI}}^\mathrm{tx}\right)$ are estimated to reconstruct the array response vectors $ \ba_\mathrm{BS}\left({\nu}_{\mathrm{IB}}^\mathrm{rx},{\xi}_{\mathrm{IB}}^\mathrm{rx}\right) $ and $ \ba_\mathrm{UE}\left({\nu}_{\mathrm{UI}}^\mathrm{tx},{\xi}_{\mathrm{UI}}^\mathrm{tx}\right) $ at the BS and UE~sides.
	\item The two effective IRS-side spatial frequencies are estimated to obtain $ [\ba_\mathrm{IRS}^\mathrm{H}\left({\nu}_{\mathrm{IB}}^\mathrm{tx},{\xi}_{\mathrm{IB}}^\mathrm{tx}\right)]_\ell
	$ $\times
	[\ba_\mathrm{IRS}\left({\nu}_{\mathrm{UI}}^\mathrm{rx},{\xi}_{\mathrm{UI}}^\mathrm{rx}\right)]_\ell $ for all $ \ell $.
	\item The overall gain $ {\gamma}_{\mathrm{IB}}{\gamma}_{\mathrm{UI}} $ common for the rank-one matrices is obtained.
	\item The remaining rank-one matrices not estimated in Step 1 are constructed by \eqref{SPAC_cascaded_prototype} using the parameters obtained from Step 2-4.
\end{enumerate}
\begin{figure}[t]
	\centering
	\includegraphics[width=0.45\textwidth]{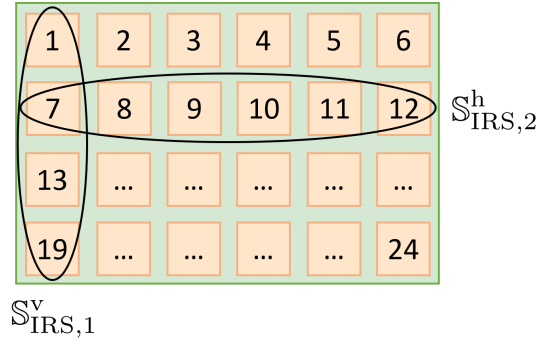} 
	\caption{An example to define the index sets for IRS elements $ \mathbb{S}_{\mathrm{IRS},x}^{\mathrm{v}} $ and $ \mathbb{S}_{\mathrm{IRS},y}^{\mathrm{h}} $ with ${L= L_\mathrm{v}\times L_\mathrm{h} = 4\times 6} $.}
	\label{SPAC_IRS_index}
\end{figure}

For clear understanding of the estimation process, we first specify the IRS element index sets for the $ x $-th column and the $ y $-th row as $ \mathbb{S}_{\mathrm{IRS},x}^{\mathrm{v}} $ and $ \mathbb{S}_{\mathrm{IRS},y}^{\mathrm{h}} $, respectively. The common sequential numbering is considered to index the IRS elements as in Fig.~\ref{SPAC_IRS_index}. With such indexing, the two index sets $ \mathbb{S}_{\mathrm{IRS},x}^{\mathrm{v}} $ and $ \mathbb{S}_{\mathrm{IRS},y}^{\mathrm{h}} $ are defined as
\begin{align}
	\mathbb{S}_{\mathrm{IRS},x}^{\mathrm{v}}&= \left\{x,L_\mathrm{h}+x,\cdots,(L_\mathrm{v}-1)L_\mathrm{h}+x\right\}, \\
	\mathbb{S}_{\mathrm{IRS},y}^{\mathrm{h}} &= \left\{(y-1)L_\mathrm{h}+1,(y-1)L_\mathrm{h}+2,\cdots,yL_\mathrm{h}\right\}.
\end{align}
In terms of the BS and UE, the UPA antenna index sets are similarly defined. 

In Step 1, the index set $ \mathbb{S}_{\mathrm{IRS}} \subset \{1,\cdots,L\}$ is defined, and the BS estimates the rank-one matrices $ {\bR}_\ell $ only for $ \ell \in \mathbb{S}_{\mathrm{IRS}} $ using the OBO estimation in Section~\ref{Tech1}. Considering the UPA structure of IRS, we employ $ \mathbb{S}_{\mathrm{IRS}} \triangleq  \mathbb{S}_{\mathrm{IRS},1}^{\mathrm{v}}\cup \mathbb{S}_{\mathrm{IRS},1}^{\mathrm{h}} $ in order that the set $ \mathbb{S}_{\mathrm{IRS}} $ contains the information of both the vertical and horizontal spatial frequencies at the IRS side. To reduce the training overhead, we let $ \mathbb{S}_{\mathrm{IRS}} $ be a small set with $ L_\mathrm{v} + L_\mathrm{h}-1$ IRS elements, while the set can include multiple columns and rows of the IRS elements. Once the rank-one matrices for $ \mathbb{S}_{\mathrm{IRS}} $ are estimated by the OBO estimation, the BS extracts the effective parameters based on the estimates $ \widehat{\bR}_\ell $ for $ \ell \in \mathbb{S}_{\mathrm{IRS}} $ to construct the remaining rank-one matrices for $ \ell \notin \mathbb{S}_{\mathrm{IRS}} $.

The spatial frequencies related to the BS and UE sides are estimated in Step 2 to reconstruct $ \ba_\mathrm{BS}\left(\nu_{\mathrm{IB}}^\mathrm{rx},\xi_{\mathrm{IB}}^\mathrm{rx}\right) $ and $ \ba_\mathrm{UE}\left(\nu_{\mathrm{UI}}^\mathrm{tx},\xi_{\mathrm{UI}}^\mathrm{tx}\right) $. In \eqref{SPAC_cascaded_prototype}, it can be seen that the column and row spaces of $ \widetilde{\bR}_\ell $ are the same as those of $ \ba_\mathrm{BS}\left(\nu_{\mathrm{IB}}^\mathrm{rx},\xi_{\mathrm{IB}}^\mathrm{rx}\right) $ and $ \ba_\mathrm{UE}^\mathrm{H}\left(\nu_{\mathrm{UI}}^\mathrm{tx},\xi_{\mathrm{UI}}^\mathrm{tx}\right) $, respectively. Therefore, we treat the left and right singular vectors corresponding to the largest singular value of the rank-one matrix $ \widehat{\bR}_\ell $ as its representative column and row. Based on the left and right singular vectors for $ \ell\in\mathbb{S}_\mathrm{IRS} $, we extract the spatial frequencies $ \left({\nu}_{\mathrm{IB}}^\mathrm{rx},{\xi}_{\mathrm{IB}}^\mathrm{rx}\right)$ and $\left({\nu}_{\mathrm{UI}}^\mathrm{tx},{\xi}_{\mathrm{UI}}^\mathrm{tx}\right)$. Focusing on the BS side and a specific $\ell  \in\mathbb{S}_{\mathrm{IRS}} $, the left singular vector $ \bu_\ell \in \mathbb{C}^{N_\mathrm{v}N_\mathrm{h}\times1} $ can be rearranged into a matrix by arranging the elements of $ \bu_\ell $ to follow the BS antenna numbering, which is similarly defined to that of the IRS in Fig. \ref{SPAC_IRS_index}. In other words, the rearranged matrix $ \boldsymbol{\mathcal{U}}_{\mathrm{BS},\ell} \in \mathbb{C}^{N_\mathrm{v}\times N_\mathrm{h}}$ can be defined as
\begin{align}\label{rearranged matrix}
	\boldsymbol{\mathcal{U}}_{\mathrm{BS},\ell} &= \begin{bmatrix}
		[\bu_\ell]_1 & \cdots & [\bu_\ell]_{N_\mathrm{h}} \\
		[\bu_\ell]_{N_\mathrm{h}+1} &  \cdots & [\bu_\ell]_{2N_\mathrm{h}} \\
		\vdots & \ddots & \vdots \\
		[\bu_\ell]_{(N_\mathrm{v}-1)N_\mathrm{h}+1} & \cdots & [\bu_\ell]_{N_\mathrm{v}N_\mathrm{h}}
	\end{bmatrix}.
\end{align}
We denote the $ x $-th column and $ y $-th row vectors of $ \boldsymbol{\mathcal{U}}_{\mathrm{BS},\ell} $ by $ \bu_{\ell,x}^\mathrm{v} \triangleq \boldsymbol{\mathcal{U}}_{\mathrm{BS},\ell}(:,x) $ and $ (\bu_{\ell,y}^\mathrm{h})^\mathrm{T} \triangleq \boldsymbol{\mathcal{U}}_{\mathrm{BS},\ell}(y,:) $.

As in \eqref{BS array response vector}, $  \ba_\mathrm{BS}\left(\nu_{\mathrm{IB}}^\mathrm{rx},\xi_{\mathrm{IB}}^\mathrm{rx}\right) $ is composed of the vertical and horizontal array response vectors. Based on the structure, the vertical spatial frequency $ \nu_{\mathrm{IB}}^\mathrm{rx} $ is estimated as
\begin{align}\label{SPAC_BS_ver}
	\widehat{\nu}_{\mathrm{IB}}^\mathrm{rx}=\frac{\sum\limits_{\ell \in \mathbb{S}_{\mathrm{IRS}}}\sum\limits_{y=2}^{N_\mathrm{v}}\boldsymbol{1}_{N_\mathrm{h}}^\mathrm{T} \left(\angle\left(\bu_{\ell,y}^\mathrm{h}\right)-\angle\left(\bu_{\ell,y-1}^\mathrm{h}\right)\right)}{ (L_\mathrm{v} + L_\mathrm{h} - 1)(N_\mathrm{v}-1)N_\mathrm{h}}.
\end{align}
Similarly, we estimate the horizontal spatial frequency $ {\xi}_{\mathrm{IB}}^\mathrm{rx} $ as
\begin{align}\label{SPAC_BS_hor}
	\widehat{\xi}_{\mathrm{IB}}^\mathrm{rx}=\frac{\sum\limits_{\ell \in \mathbb{S}_{\mathrm{IRS}}}\sum\limits_{x=2}^{N_\mathrm{h}}\boldsymbol{1}_{N_\mathrm{v}}^\mathrm{T} \left(\angle\left(\bu_{\ell,x}^\mathrm{v}\right)-\angle\left(\bu_{\ell,x-1}^\mathrm{v}\right)\right)}{ (L_\mathrm{v} + L_\mathrm{h} - 1)(N_\mathrm{h}-1)N_\mathrm{v}}.
\end{align}
In words, the estimates in \eqref{SPAC_BS_ver} and \eqref{SPAC_BS_hor} are the sample averages of spatial frequencies based on \eqref{rearranged matrix}. Now, the estimate of BS-side array response vector $ \ba_\mathrm{BS}\left({\nu}_{\mathrm{IB}}^\mathrm{rx},{\xi}_{\mathrm{IB}}^\mathrm{rx}\right) $ is reconstructed as in \eqref{BS array response vector} with the two estimated spatial frequencies $ \widehat{\nu}_{\mathrm{IB}}^\mathrm{rx}$ and $\widehat{\xi}_{\mathrm{IB}}^\mathrm{rx} $. Using the right singular vectors $ \bv_\ell $ for $ \ell \in \mathbb{S}_{\mathrm{IRS}} $, the UE-side array response vector is similarly estimated as $ \ba_\mathrm{UE}\left(\widehat{\nu}_{\mathrm{UI}}^\mathrm{tx},\widehat{\xi}_{\mathrm{UI}}^\mathrm{tx}\right) $ by deriving $ \widehat{\nu}_{\mathrm{UI}}^\mathrm{tx} $ and $ \widehat{\xi}_{\mathrm{UI}}^\mathrm{tx} $ as~in~\eqref{SPAC_BS_ver}~and~\eqref{SPAC_BS_hor}.

In Step 3, we define the effective two IRS-side spatial frequencies $ \nu_{\mathrm{IRS}} $ and $ \xi_{\mathrm{IRS}} $ as
	\begin{align}
		\nu_{\mathrm{IRS}} = \nu_{\mathrm{UI}}^\mathrm{rx}-\nu_{\mathrm{IB}}^\mathrm{tx}, \enspace
		\xi_{\mathrm{IRS}} = \xi_{\mathrm{UI}}^\mathrm{rx}-\xi_{\mathrm{IB}}^\mathrm{tx},
	\end{align}
	which are estimated instead of each of four spatial frequencies. To explain why this is possible, based on the single-path approximations in \eqref{SPAC_IB} and \eqref{SPAC_UI}, we have
\begingroup\allowdisplaybreaks
\begin{align}
	c_\ell &= \ba_\mathrm{BS}^\mathrm{H}\left(\nu_{\mathrm{IB}}^\mathrm{rx},\xi_{\mathrm{IB}}^\mathrm{rx}\right) \widetilde{\bR}_\ell\ba_\mathrm{UE}\left(\nu_{\mathrm{UI}}^\mathrm{tx},\xi_{\mathrm{UI}}^\mathrm{tx}\right)\notag\\
	&=\gamma_{\mathrm{IB}}\gamma_{\mathrm{UI}}[\ba_\mathrm{IRS}^\mathrm{H}\left(\nu_{\mathrm{IB}}^\mathrm{tx},\xi_{\mathrm{IB}}^\mathrm{tx}\right)]_\ell [\ba_\mathrm{IRS}\left(\nu_{\mathrm{UI}}^\mathrm{rx},\xi_{\mathrm{UI}}^\mathrm{rx}\right)]_\ell
	\notag\\	&=
	\gamma_{\mathrm{IB}}\gamma_{\mathrm{UI}}[\ba_\mathrm{IRS}^*\left(\nu_{\mathrm{IB}}^\mathrm{tx},\xi_{\mathrm{IB}}^\mathrm{tx}\right) \odot \ba_\mathrm{IRS}\left(\nu_{\mathrm{UI}}^\mathrm{rx},\xi_{\mathrm{UI}}^\mathrm{rx}\right)]_\ell,\label{SPAC_IRS}
\end{align}
for each $ \ell $. This clearly shows that we only need to estimate $ \ba_\mathrm{IRS}^*\left(\nu_{\mathrm{IB}}^\mathrm{tx},\xi_{\mathrm{IB}}^\mathrm{tx}\right) \odot \ba_\mathrm{IRS}\left(\nu_{\mathrm{UI}}^\mathrm{rx},\xi_{\mathrm{UI}}^\mathrm{rx}\right) $ to construct $ \widetilde{\bR}_\ell $ for $ \ell \notin \mathbb{S}_{\mathrm{IRS}} $ since $ \widetilde{\bR}_\ell = \ba_\mathrm{BS}\left(\nu_{\mathrm{IB}}^\mathrm{rx},\xi_{\mathrm{IB}}^\mathrm{rx}\right)c_\ell \ba_\mathrm{UE}^\mathrm{H}\left(\nu_{\mathrm{UI}}^\mathrm{tx},\xi_{\mathrm{UI}}^\mathrm{tx}\right)$ where the two array response vectors in the left and right are estimated in Step 2. The two gains $ \gamma_{\mathrm{IB}}$ and $\gamma_{\mathrm{UI}} $ in \eqref{SPAC_IRS} also need to be estimated, which will be handled in Step 4. Considering the structure of array response vector at the IRS side in \eqref{Approx_arv_IRS}, the Hadamard product of the two vectors in~\eqref{SPAC_IRS} is expressed as 
\begin{align}
	&\ba_\mathrm{IRS}^*\left(\nu_{\mathrm{IB}}^\mathrm{tx},\xi_{\mathrm{IB}}^\mathrm{tx}\right) \odot \ba_\mathrm{IRS}\left(\nu_{\mathrm{UI}}^\mathrm{rx},\xi_{\mathrm{UI}}^\mathrm{rx}\right)\notag \\
	= &\frac{1}{\sqrt{L}}\left(\begin{bmatrix}
		1 \\ \vdots \\ e^{-j(L_\mathrm{v}-1)\nu_{\mathrm{IB}}^\mathrm{tx}} 
	\end{bmatrix} \otimes \begin{bmatrix}
		1 \\ \vdots \\ e^{-j(L_\mathrm{h}-1)\xi_{\mathrm{IB}}^\mathrm{tx}}
	\end{bmatrix}  \right) 
	\notag \\	&
	\odot \frac{1}{\sqrt{L}}\left(\begin{bmatrix}
		1 \\ \vdots \\ e^{j(L_\mathrm{v}-1)\nu_{\mathrm{UI}}^\mathrm{rx}} 
	\end{bmatrix} \otimes \begin{bmatrix}
		1 \\ \vdots \\ e^{j(L_\mathrm{h}-1)\xi_{\mathrm{UI}}^\mathrm{rx}}
	\end{bmatrix}  \right) \notag\\
	\stackrel{(\mathrm{a})}{=}&\frac{1}{\sqrt{L}}\left(\begin{bmatrix}
		1 \\ \vdots \\ e^{-j(L_\mathrm{v}-1)\nu_{\mathrm{IB}}^\mathrm{tx}} 
	\end{bmatrix} \odot \begin{bmatrix}
		1 \\ \vdots \\ e^{j(L_\mathrm{v}-1)\nu_{\mathrm{UI}}^\mathrm{rx}} 
	\end{bmatrix}  \right) 
	\notag \\	&
	\otimes \frac{1}{\sqrt{L}}\left(\begin{bmatrix}
		1 \\ \vdots \\ e^{-j(L_\mathrm{h}-1)\xi_{\mathrm{IB}}^\mathrm{tx}}
	\end{bmatrix} \odot \begin{bmatrix}
		1 \\ \vdots \\ e^{j(L_\mathrm{h}-1)\xi_{\mathrm{UI}}^\mathrm{rx}}
	\end{bmatrix}  \right) \notag\\
	=& \frac{1}{\sqrt{L}}\left[1,\cdots,e^{j(L_\mathrm{v}-1)\left(\nu_{\mathrm{UI}}^\mathrm{rx}-\nu_{\mathrm{IB}}^\mathrm{tx}\right)} \right]^\mathrm{T} 
	\notag \\	&
	\otimes\frac{1}{\sqrt{L}}\left[ 1,\cdots,e^{j(L_\mathrm{h}-1)(\xi_{\mathrm{UI}}^\mathrm{rx}-\xi_{\mathrm{IB}}^\mathrm{tx})} \right]^\mathrm{T} \notag\\
	=& \frac{1}{\sqrt{L}}\ba_\mathrm{IRS}\left(\nu_{\mathrm{UI}}^\mathrm{rx}-\nu_{\mathrm{IB}}^\mathrm{tx},\xi_{\mathrm{UI}}^\mathrm{rx}-\xi_{\mathrm{IB}}^\mathrm{tx}\right)\notag\\
	=& \frac{1}{\sqrt{L}}\ba_\mathrm{IRS}\left(\nu_{\mathrm{IRS}},\xi_{\mathrm{IRS}}\right), \label{SPAC_Step 3}
\end{align}
\endgroup
where (a) is based on the property that $ \left(\bA \otimes \bB\right) \odot \left(\bC \otimes \bD\right) =\left(\bA \odot \bC\right) \otimes \left(\bB \odot \bD\right) $. This implies that only the two effective spatial frequencies $ \nu_{\mathrm{IRS}} $ and $ \xi_{\mathrm{IRS}} $ are needed to construct $ \ba_\mathrm{IRS}^*\left(\nu_{\mathrm{IB}}^\mathrm{tx},\xi_{\mathrm{IB}}^\mathrm{tx}\right) \odot \ba_\mathrm{IRS}\left(\nu_{\mathrm{UI}}^\mathrm{rx},\xi_{\mathrm{UI}}^\mathrm{rx}\right) $.

To estimate the two spatial frequencies $ \nu_{\mathrm{IRS}} $ and $ \xi_{\mathrm{IRS}} $, we can exploit the actual observation of $ c_\ell $ in the form of ${\widehat{c}_\ell \triangleq \ba_\mathrm{BS}^\mathrm{H}\left(\widehat{\nu}_{\mathrm{IB}}^\mathrm{rx},\widehat{\xi}_{\mathrm{IB}}^\mathrm{rx}\right) \widehat{\bR}_\ell\ba_\mathrm{UE}\left(\widehat{\nu}_{\mathrm{UI}}^\mathrm{tx},\widehat{\xi}_{\mathrm{UI}}^\mathrm{tx}\right)} $ for $ \ell \in \mathbb{S}_{\mathrm{IRS}} $ with the parameters obtained in Steps 1 and 2. For $ L_\mathrm{v} $ observations of $ \widehat{c}_\ell $ for $ \ell \in \mathbb{S}_{\mathrm{IRS},1}^{\mathrm{v}} $, the estimated vertical spatial frequency $ \widehat{\nu}_\mathrm{IRS} $~is
\begin{align}
	\widehat{\nu}_\mathrm{IRS}=\frac{1}{L_\mathrm{v}-1} \sum_{\substack{\ell \neq 1 \\ \ell\in \mathbb{S}_{\mathrm{IRS},1}^{\mathrm{v}}}}\angle \left(\frac{\widehat{c}_{\ell}}{\widehat{c}_{\ell-1}}\right).
\end{align}
For $ L_\mathrm{h} $ observations for $ \ell \in \mathbb{S}_{\mathrm{IRS},1}^{\mathrm{h}} $, the horizontal spatial frequency is estimated as
\begin{align}
	\widehat{\xi}_\mathrm{IRS}=\frac{1}{L_\mathrm{h}-1} \sum_{\substack{\ell \neq 1 \\ \ell\in \mathbb{S}_{\mathrm{IRS},1}^{\mathrm{h}}}}\angle \left(\frac{\widehat{c}_{\ell}}{\widehat{c}_{\ell-1}}\right).
\end{align}
With the estimated spatial frequencies, the BS constructs the IRS-side array~response vector~in~\eqref{SPAC_Step 3}.

The overall gain $ \gamma_\mathrm{IRS} \triangleq {\gamma}_{\mathrm{IB}}{\gamma}_{\mathrm{UI}} $ is estimated in Step 4 instead of each gain separately. Using \eqref{SPAC_IRS} and \eqref{SPAC_Step 3} in Step~3, the overall gain can be directly given as
\begin{align}
	\label{SPAC_gain}
	{\gamma}_{\mathrm{IRS}}&=\frac{{c}_\ell}{\frac{1}{\sqrt{L}}[\ba_{\mathrm{IRS}}\left({\nu}_\mathrm{IRS},{\xi}_\mathrm{IRS}\right)]_\ell}
	\notag\\	&
	=\frac {\gamma_{\mathrm{IB}}\gamma_{\mathrm{UI}}[\ba_\mathrm{IRS}^*\left(\nu_{\mathrm{IB}}^\mathrm{tx},\xi_{\mathrm{IB}}^\mathrm{tx}\right) \odot \ba_\mathrm{IRS}\left(\nu_{\mathrm{UI}}^\mathrm{rx},\xi_{\mathrm{UI}}^\mathrm{rx}\right)]_\ell }{\frac{1}{\sqrt{L}}[\ba_{\mathrm{IRS}}\left({\nu}_\mathrm{IRS},{\xi}_\mathrm{IRS}\right)]_\ell}
	\notag\\	&
	=\gamma_{\mathrm{IB}}\gamma_{\mathrm{UI}},
\end{align}
which can be obtained for any $ \ell \in \mathbb{S}_{\mathrm{IRS}} $. By utilizing $ \widehat{c}_\ell $, $ \widehat{\nu}_\mathrm{IRS} $, and $ \widehat{\xi}_\mathrm{IRS} $ for $ \ell \in \mathbb{S}_{\mathrm{IRS}} $ obtained in Step 3, $ L_\mathrm{v} + L_\mathrm{h} - 1 $ observations of $ {\gamma}_{\mathrm{IRS}} $ can be computed as in \eqref{SPAC_gain}. Based on the observations, the effective gain is estimated as
\begin{align}
	\widehat{\gamma}_{\mathrm{IRS}} = \frac{1}{ L_\mathrm{v} + L_\mathrm{h} - 1} \sum_{\ell \in \mathbb{S}_{\mathrm{IRS}}} \frac{\widehat{c}_\ell}{\frac{1}{\sqrt{L}}\left[\ba_\mathrm{IRS}\left(\widehat{\nu}_\mathrm{IRS},\widehat{\xi}_\mathrm{IRS}\right)\right]_\ell}.
\end{align}
Now we can reconstruct the remaining rank-one matrices $ \widehat{\bR}_\ell $ for $ \ell \notin \mathbb{S}_{\mathrm{IRS}} $ by using all estimated parameters as
\begin{align}
	\label{SPAC_cascaded}
	&\widehat{\bR}_\ell \notag\\	
	=& \frac{\widehat{\gamma}_\mathrm{IRS}}{\sqrt{L}} \ba_\mathrm{BS}\left(\widehat{\nu}_{\mathrm{IB}}^\mathrm{rx},\widehat{\xi}_{\mathrm{IB}}^\mathrm{rx}\right)\left[\ba_\mathrm{IRS}\left(\widehat{\nu}_\mathrm{IRS},\widehat{\xi}_\mathrm{IRS}\right)\right]_\ell \ba_\mathrm{UE}^\mathrm{H}\left(\widehat{\nu}_{\mathrm{UI}}^\mathrm{tx},\widehat{\xi}_{\mathrm{UI}}^\mathrm{tx}\right).
\end{align}

\begingroup\allowdisplaybreaks
As the rank-one matrix estimation utilizing uplink signaling is conducted only for $ \ell \in \mathbb{S}_{\mathrm{IRS}} $, the training overhead for SPAC is $ \tau_\mathrm{c}=\tau_{\mathrm{SPAC}}=(L_\mathrm{v} + L_\mathrm{h} - 1)M $. Compared with the training overhead of the OBO estimation $ \tau_{\mathrm{OBO}}=LM=L_\mathrm{v}L_\mathrm{h}M $, the overhead of SPAC is remarkably low especially with large $ L $. With the single-path channel approximation, SPAC substitutes the problem of large dimensional channel estimation into that of the small number of parameter estimations, contributing to low training overhead.

\subsection{Selective emphasis on rank-one matrices (SEROM)}
\label{SEROM}
SEROM is proposed to conduct efficient channel estimation with the design of IRS reflection-coefficient {matrices}. Different from the previous techniques, SEROM always turns on the entire IRS elements and utilizes both the IRS phase shifts and uplink signaling for channel estimation. We first denote the IRS reflection-coefficient matrix for the $ q $-th training period by $ \boldsymbol{\Phi}^{(q)} $, which is defined as
\endgroup
\begin{align}
	\boldsymbol{\Phi}^{(q)} = \diag\left(\left[e^{j\phi_1^{(q)}},\cdots,e^{j\phi_L^{(q)}}\right]^{\mathrm{T}}\right),
\end{align}
with $ q\in\{1,\cdots,Q\} $ where $ Q $ is the total number of training periods. The IRS reflection-coefficient matrix $ \boldsymbol{\Phi}[t] $ is fixed as $ \boldsymbol{\Phi}^{(q)} $ during the $ q $-th training period $ \tau_{\mathrm{d}}+(q-1)M+1 \leq t \leq \tau_{\mathrm{d}}+qM $. The length of each training period for SEROM is $ M $, which is equal for the OBO estimation and SPAC. However, the number of training periods for the two previous techniques is $ L $ and  $ L_\mathrm{v} + L_\mathrm{h}-1 $, and it implies that their training overhead depends on the number of IRS elements. SEROM can adapt the training overhead flexibly since $ Q $ is the adjustable parameter independent \red{of a} system structure.

As in Section \ref{Tech1}, the UE transmits the length $ M $ training sequence with $ s_\mathrm{UL}[t]=\sqrt{P_\mathrm{UL}} $ and exploits the normalized $ M\times M $ DFT matrix as the training beamformer $ \bF_{\mathrm{UIB},q} =[\bff[\tau_{\mathrm{d}}+(q-1)M+1],\cdots,\bff[\tau_{\mathrm{d}}+qM]] $ for~each~$ q $.
\begingroup\allowdisplaybreaks
The BS processes the $ M $ received signals as
\begin{align}
	\label{SEROM_ZF}
	&\frac{1}{\sqrt{P_\mathrm{UL}}}\widetilde{\bY}_{\mathrm{UIB},q}\bF_{\mathrm{UIB},q}^\mathrm{H} 
	\notag\\	=& 
	\bH_\mathrm{IB}\boldsymbol{\Phi}^{(q)}\bH_\mathrm{UI} +\frac{1}{\sqrt{P_\mathrm{UL}}}\widetilde{\bN}_{\mathrm{UIB},q}\bF_{\mathrm{UIB},q}^\mathrm{H},
\end{align}
for the $ q $-th training period. Recalling that $ \bH_\mathrm{IB}\boldsymbol{\Phi}^{(q)}\bH_\mathrm{UI} = \sum_{\ell=1}^L e^{j\phi_\ell^{(q)}}\bR_\ell  $ as in \eqref{one_rank_sum}, the cascaded UE-IRS-BS channel can be expressed as
\begin{align}
	\bH_\mathrm{IB}\boldsymbol{\Phi}^{(q)}\bH_\mathrm{UI}
	=&
	\begin{bmatrix}
		e^{j\phi_1^{(q)}}\bI_N&\cdots&e^{j\phi_L^{(q)}}\bI_N 
	\end{bmatrix}
	\begin{bmatrix}
		\bR_1 \\ \vdots \\ \bR_L
	\end{bmatrix} 
	\notag\\	=&
	\left( \begin{bmatrix}
		e^{j\phi_1^{(q)}}&\cdots&e^{j\phi_L^{(q)}}
	\end{bmatrix}\otimes \bI_N \right)
	\begin{bmatrix}
		\bR_1 \\ \vdots \\ \bR_L
	\end{bmatrix}.\label{SEROM_cascaded}
\end{align}
Then, we can stack the cascaded channel through the IRS $ \bH_\mathrm{IB}\boldsymbol{\Phi}^{(q)}\bH_\mathrm{UI} $ as
\endgroup
\begin{align}\label{Omega}
	&\begin{bmatrix}
		\bH_\mathrm{IB}\boldsymbol{\Phi}^{(1)}\bH_\mathrm{UI} \\ \vdots \\ \bH_\mathrm{IB}\boldsymbol{\Phi}^{(Q)}\bH_\mathrm{UI}
	\end{bmatrix}
	\notag\\		=& 
	\left(
	\underbrace{\begin{bmatrix}
			e^{j\phi_1^{(1)}}&\cdots&e^{j\phi_L^{(1)}} \\
			\vdots& \ddots& \vdots \\
			e^{j\phi_1^{(Q)}}&\cdots&e^{j\phi_L^{(Q)}} \\
	\end{bmatrix}}_{{\triangleq \boldsymbol{\Omega}}}\otimes\text{ } \bI_N \right)
	\begin{bmatrix}
		\bR_1 \\ \vdots \\ \bR_L
	\end{bmatrix},
\end{align}
where $ \boldsymbol{\Omega} \in \mathbb{C}^{Q\times L} $ is the IRS training matrix, whose elements are unit modulus.

The IRS training matrix $ \boldsymbol{\Omega} $ in \eqref{Omega} can be designed to have mutually orthogonal columns for the product $ \big(\boldsymbol{\Omega}^\mathrm{H}\otimes \bI_N\big)\big(\boldsymbol{\Omega}\otimes \bI_N\big)=\big(\boldsymbol{\Omega}^\mathrm{H}\boldsymbol{\Omega}\otimes \bI_N\big) $ to be a non-zero diagonal matrix. This condition facilitates perfect extraction of the rank-one matrices from the stacked cascaded channel in \eqref{Omega}. However, it is feasible only when the number of training periods $ Q $ is larger than or equal to the number of the IRS elements $ L $. For large $ L $, which is typical for IRS-empowered systems, a number of training periods are required to satisfy such orthogonality, and this motivates us to design the IRS training matrix under the condition $ Q < L $.

Since it is impossible to make the columns of $ \boldsymbol{\Omega} $ mutually orthogonal for $ Q < L $, we design the IRS training matrix to have pseudo-orthogonal columns as
\begin{align}
	(\boldsymbol{\Omega}(:,\ell))^{\mathrm{H}}\boldsymbol{\Omega}(:,k) = \begin{cases}
		a_\ell, & \text{for } \ell=k, \\
		b_{\ell,k}, & \text{otherwise},
	\end{cases}
\end{align}
satisfying $ \lvert a_\ell \rvert \gg \lvert b_{\ell,k} \rvert $ for all $ \ell $ and $ k $. To design such $ \boldsymbol{\Omega} $, we can employ a submatrix by choosing $ Q $ rows for $ Q<L $ or $ L $ columns for $ Q\geq L $ from the $ \cM \times \cM $ DFT matrix where $ \cM = \max\{Q,L\} $. The BS finally conducts the rank-one matrix estimation~as
\begin{align}\label{SEROM_one_rank}
	\begin{bmatrix}
		\widehat{\bR}_1 \\ \vdots \\ \widehat{\bR}_L
	\end{bmatrix} 
	=&\frac{A}{\sqrt{P_\mathrm{UL}}}\left(\boldsymbol{\Omega}^\mathrm{H}\otimes \bI_N\right)
	\begin{bmatrix}\widetilde{\bY}_{\mathrm{UIB},1}\bF_{\mathrm{UIB},1}^\mathrm{H}\\ \vdots \\
		\widetilde{\bY}_{\mathrm{UIB},Q}\bF_{\mathrm{UIB},Q}^\mathrm{H}
	\end{bmatrix} \notag\\
	=& A\left(
	\begin{bmatrix}
		a_1&\cdots&b_{1,L} \\
		\vdots& \ddots& \vdots \\
		b_{L,1}&\cdots&a_L \\
	\end{bmatrix}\otimes \bI_N \right)
	\begin{bmatrix}
		\bR_1 \\ \vdots \\ \bR_L
	\end{bmatrix} 
	\notag \\		&
	+ \frac{A}{\sqrt{P_\mathrm{UL}}}\left(\boldsymbol{\Omega}^\mathrm{H}\otimes \bI_N\right)
	\begin{bmatrix}\widetilde{\bN}_{\mathrm{UIB},1}\bF_{\mathrm{UIB},1}^\mathrm{H} \\ \vdots \\
		\widetilde{\bN}_{\mathrm{UIB},Q}\bF_{\mathrm{UIB},Q}^\mathrm{H}
	\end{bmatrix}.
\end{align}
The normalization factor $A$ to cancel the amplification effect of $ \boldsymbol{\Omega}^{\mathrm{H}}\boldsymbol{\Omega} $ is defined as
\begingroup\allowdisplaybreaks
\begin{align}
	A &= \frac{\sum_{l=1}^L (\boldsymbol{\Omega}(:,\ell))^{\mathrm{H}}\boldsymbol{\Omega}(:,\ell)}{Q\sum_{l=1}^L \lvert \sum_{k=1}^L (\boldsymbol{\Omega}(:,\ell))^{\mathrm{H}}\boldsymbol{\Omega}(:,k) \rvert}
	\notag\\	&
	= \frac{\sum_{l=1}^L a_\ell}{Q\sum_{l=1}^L \lvert a_\ell + \sum_{k\neq l} b_{\ell,k} \rvert}
	\notag\\	&
	= \frac{L}{\sum_{l=1}^L \lvert Q + \sum_{k\neq l} b_{\ell,k} \rvert}, \label{Normalization factor}
\end{align}
where $ a_\ell = Q $ holds for all $ \ell $ since the entire IRS elements are turned on with the unit modulus constraint. For $ {Q \geq L} $, the $ L $ columns of $ {Q \times Q} $ DFT matrix can be chosen to give $ {b_{\ell,k}=0} $ and $ {A=1/Q} $. For $ Q<L $, the $ Q $-row submatrix from the $ {\cM \times \cM} $ DFT matrix can be chosen to satisfy $ {\lvert a_\ell \rvert \gg \lvert b_{\ell,k} \rvert} $ and $ {A\approx 1/Q} $.

The overall training overhead of SEROM is $ \tau_\mathrm{c}=\tau_\mathrm{SEROM}= QM $. Note that $ \tau_\mathrm{SEROM} $ is independent from the number of IRS elements $ L $. For the small number of the IRS elements, we can take $ Q\geq L $ with moderate training overhead, and the IRS training matrix with $ A=1/Q $ ensures perfect rank-one matrix estimation in \eqref{SEROM_one_rank} at noiseless circumstance. However, keeping the condition $ Q\geq L $ makes the minimum length of training sequences proportional to $ L $, which is undesirable for typical IRS-empowered systems adopting large $ L $. In this case, we can set $ Q < L $ or even $ Q \ll L $ to suppress the training overhead in a moderate range.
\endgroup


\section{IRS Phase Shift Design} \label{IRS Beamforming Techniques}
The considered IRS-empowered SU-MIMO system is intended to serve the~UE with high spectral efficiency through the support of the IRS. In this section, we propose a novel phase shift design at the IRS to achieve high spectral efficiency. It can be shown that the proposed design gives an optimal phase shift that maximizes the spectral efficiency for each IRS element while the phase shifts of other IRS elements are fixed. In addition, all the processes require only basic linear matrix operations making the proposed design practical. We first assume perfect channel information at the BS for conceptual explanation. Then, for the numerical results in Section~\ref{numerical results}, we examine the proposed phase shift design with the perfect channel information and also with the estimated channels by the proposed techniques in Section~\ref{Channel Estimation Techniques}.  

\subsection{Optimal phase shift for each IRS element}
Relying on the downlink and uplink channel reciprocity in TDD \cite{Larsson:2014_TDD}, we take the conjugate transpose to represent the total downlink channel $\bH_\mathrm{tot}^\mathrm{H}$ where the total channel $\bH_\mathrm{tot}$ is represented~by
\begin{align}
	\bH_{\mathrm{tot}}=\bH_\mathrm{UB}+\bH_{\mathrm{IB}}\boldsymbol{\Phi}[t]\bH_{\mathrm{UI}}=\bH_\mathrm{UB}+\sum_{\ell=1}^{L}e^{j \phi_{\ell}[t]} \bR_{\ell}. \label{total downlink channel}
\end{align}
Then, the downlink spectral efficiency $R_\mathrm{DL}$ is given as \cite{spectral_efficiency} 
\begin{align}
	R_\mathrm{DL}&=\log_2\left( \det\left(\bI_r + \frac{P_\mathrm{DL}}{r N_0}\bW^\mathrm{H}\bH_\mathrm{tot}\bH_\mathrm{tot}^\mathrm{H}\bW\right)\right), \label{spectral efficiency} 
\end{align}
where $r$ is the rank of total downlink channel $\bH_\mathrm{tot}^\mathrm{H}$, and $\bW\in\mathbb{C}^{N \times r}$ is the downlink transmit beamformer at the BS. Since $\boldsymbol{\Phi}[t]$ is designed based on given channels, $\boldsymbol{\Phi}[t]$~is fixed during the data transmissions, omitting the time index $t$ as $\boldsymbol{\Phi}$. We turn on all the IRS elements, i.e., $\boldsymbol{\Phi}=\diag\left(\left[e^{j \phi_1},\cdots,e^{j\phi_L}\right]^\mathrm{T}\right)$, to maximize the reflected signal strengths. 

With the given $\boldsymbol{\Phi}$ and $\bH_\mathrm{tot}^\mathrm{H}$, the beamformer $\bW$ is given as the dominant $r$ right singular vectors of $\bH_\mathrm{tot}^\mathrm{H}$ as \cite{MIMO_SVD_1,MIMO_SVD_2}
\begin{align}
	\bW&=\bV_\mathrm{tot}(:,1:r), \label{transmit beamformer} \\
	\bH_\mathrm{tot}^\mathrm{H}&=\bU_\mathrm{tot} \boldsymbol{\Sigma}_\mathrm{tot}\bV_\mathrm{tot}^\mathrm{H}, \label{svd}
\end{align}
where (\ref{svd}) is the singular value decomposition (SVD) of $\bH_\mathrm{tot}^\mathrm{H}$. 
On one hand, $\bH_{\mathrm{tot}}^\mathrm{H}$ contains $\boldsymbol{\Phi}$ as in (\ref{total downlink channel}), which let $\bW$ depend on $\boldsymbol{\Phi}$. On the other hand, the design of $\boldsymbol{\Phi}$ that is to maximize $R_\mathrm{DL}$ in (\ref{spectral efficiency}) also depends on $\bW$. This entangled correlation of $\boldsymbol{\Phi}$ and $\bW$ makes it difficult to jointly design the optimal $\boldsymbol{\Phi}$ and $\bW$. Hence, we first reformulate $R_{\mathrm{DL}}$ in~(\ref{spectral efficiency}) to decompose the design of $\boldsymbol{\Phi}$ and $\bW$ by exploiting the property between $\bH_\mathrm{tot}^\mathrm{H}$ and $\bW$ in (\ref{transmit beamformer}) as
\begin{align}
	R_\mathrm{DL}=&\log_2\bigg(\det\bigg(\bI_r + \frac{P_\mathrm{DL}}{r N_0}\bV_\mathrm{tot}^\mathrm{H}(:,1:r)\bH_\mathrm{tot} 
	\notag \\	&
	\times\bH_\mathrm{tot}^\mathrm{H}\bV_\mathrm{tot}(:,1:r)\bigg)\bigg) \notag \\
	\stackrel{(\mathrm{a})}{=}&\log_2\left( \det\left(\bI_N + \frac{P_\mathrm{DL}}{r N_0}\bV_\mathrm{tot}^\mathrm{H}\bH_\mathrm{tot}\bH_\mathrm{tot}^\mathrm{H}\bV_\mathrm{tot}\right)\right) \notag \\  
	=&\log_2\left( \det\left(\bV_\mathrm{tot}^\mathrm{H}\left(\bI_N + \frac{P_\mathrm{DL}}{r N_0}\bH_\mathrm{tot}\bH_\mathrm{tot}^\mathrm{H}\right)\bV_\mathrm{tot}\right)\right) \notag \\
	\stackrel{(\mathrm{b})}{=}&\log_2\left( \det\left(\bI_N + \frac{P_\mathrm{DL}}{r N_0}\bH_\mathrm{tot}\bH_\mathrm{tot}^\mathrm{H}\right)\right), \label{spectral efficiecny w/o beamformer}
\end{align}
where (a) holds since the rank of $\bH_{\mathrm{tot}}^\mathrm{H}$ is given by $r$, and (b) holds with the fact that $\det(\bA\bB)=\det(\bA)\det(\bB)$ for any square matrices $\bA$ and $\bB$ with the same dimension and that $\bV_\mathrm{tot}$ is a unitary matrix. 
The reformulated spectral efficiency $R_\mathrm{DL}$ in (\ref{spectral efficiecny w/o beamformer}) is independent from the specific value of $\bW$. This allows to design $\boldsymbol{\Phi}$ first to maximize $R_\mathrm{DL}$. Then, $\bW$ can be designed as in (\ref{transmit beamformer}) with the designed~$\boldsymbol{\Phi}$ and downlink channel~$\bH_{\mathrm{tot}}^\mathrm{H}$.

\begingroup\allowdisplaybreaks
To get the optimal value of the $\ell$-th phase shift $\phi_\ell$ that maximizes $R_\mathrm{DL}$ in (\ref{spectral efficiecny w/o beamformer}) for given $\{\phi_k\}_{k=1, k\neq \ell}^{L}$, we set the optimization problem as
\begin{align}
	\max_{\phi_\ell} \enspace \det \left( \bI_N + \lambda \left( \bH_{-\ell}+e^{j\phi_\ell} \bR_{\ell}\right)\left( \bH_{-\ell}+e^{j\phi_\ell} \bR_{\ell}\right)^\mathrm{H} \right), \label{IRS phase optimization problem}
\end{align}
where $\lambda=P_\mathrm{DL}/(r N_0)$, and $\bH_{-\ell}$$=\bH_\mathrm{UB}+\sum_{k=1,k\neq \ell}^{L}e^{j\phi_k} \bR_{k}$, which gives $\bH_\mathrm{tot}=\bH_{-\ell}+e^{j\phi_\ell} \bR_{\ell}$. By substituting $\bR_{\ell}$$=\bh_{\mathrm{IB},\ell}\bh_{\mathrm{UI},\ell}^\mathrm{H}$, the objective function in (\ref{IRS phase optimization problem}) can be reformulated~as
\begin{align}
	&\det \left( \bI_N + \lambda \left( \bH_{-\ell}+e^{j\phi_\ell} \bR_{\ell}\right)\left( \bH_{-\ell}+e^{j\phi_\ell} \bR_{\ell}\right)^\mathrm{H} \right) \notag \\
	=&\det \bigg(\bI_N+\lambda\bigg( \bH_{-\ell}\bH_{-\ell}^\mathrm{H}+e^{j\phi_{\ell}}\bh_{\mathrm{IB},\ell}\Big(\bH_{-\ell}\bh_{\mathrm{UI},\ell}\Big)^\mathrm{H} \notag \\           &+e^{-j\phi_\ell}\Big(\bH_{-\ell}\bh_{\mathrm{UI},\ell}\Big)\bh_{\mathrm{IB},\ell}^\mathrm{H}+\bh_{\mathrm{IB},\ell}\bh_{\mathrm{UI},\ell}^\mathrm{H}\Big(\bh_{\mathrm{IB},\ell}\bh_{\mathrm{UI},\ell}^\mathrm{H}\Big)^\mathrm{H}\bigg)\bigg). \label{phi full expansion}
\end{align}
For simplicity, let us define the following variables:
\begin{align}
	\kappa_\ell=&e^{j\phi_\ell}\lambda, \label{kappa definition} \\
	\bp_\ell=&\bh_{\mathrm{IB},\ell}, \\
	\bq_\ell=&\bH_{-\ell}\bh_{\mathrm{UI},\ell}, \\
	\bA_\ell=&\bI_N+\lambda \bigg(\bH_{-\ell}\bH_{-\ell}^\mathrm{H}+\bh_{\mathrm{IB},\ell}\bh_{\mathrm{UI},\ell}^\mathrm{H}\Big(\bh_{\mathrm{IB},\ell}\bh_{\mathrm{UI},\ell}^\mathrm{H}\Big)^\mathrm{H}\bigg) \label{A_ell definition}.
\end{align}
By using these variables, (\ref{phi full expansion}) can be represented as
\endgroup
\begin{align}
	&\det \Big( \bA_\ell +\kappa_\ell \bp_\ell \bq_\ell^\mathrm{H} + \kappa_\ell^* \bq_\ell \bp_\ell^\mathrm{H} \Big) \notag \\
	=&\det \left( \bA_\ell + \left[\bp_\ell, \bq_\ell \right] \diag\left(\left[\kappa_\ell,\kappa_\ell^*\right]^\mathrm{T}\right)\left[\bq_\ell,\bp_\ell\right]^\mathrm{H} \right) \notag\\
	\stackrel{(\mathrm{a})}{=}&\det \left(\diag\left(\left[\frac{1}{\kappa_\ell},\frac{1}{\kappa_\ell^*}\right]^\mathrm{T}\right) +\left[\bq_\ell,\bp_\ell\right]^\mathrm{H} \bA_\ell^{-1} \left[\bp_\ell, \bq_\ell \right]\right) 
	\notag \\	& \times
	\det \left(\diag\left([\kappa_\ell,\kappa_\ell^*]^\mathrm{T}\right)\right)\det(\bA_\ell), \label{after Sylvester's det thm}
\end{align}
where (a) can be derived using the Sylvester's determinant theorem \cite{matrix_book_2}.
The existence of $\bA_\ell^{-1}$ in (\ref{after Sylvester's det thm}) can be proven by the following lemma using the structure of $\bA_\ell$ in (\ref{A_ell definition}).
\begin{lemma} \label{lemma-inverse of A}
	For any positive definite matrix $\bA$ and a matrix~$\bB$ with a proper dimension, $\bA+\bB\bB^\mathrm{H}$ is an invertible matrix.
\end{lemma}
\begin{proof} 
	Suppose that $\bx$ is any non-zero vector. Then, we have
	\begin{align}
		\bx^\mathrm{H}\left(\bA+\bB\bB^\mathrm{H}\right)\bx
		&=\bx^\mathrm{H}\bA\bx+\bx^\mathrm{H}\bB\bB^\mathrm{H}\bx
		\notag \\&
		=\bx^\mathrm{H}\bA\bx+\left\Vert\bB^\mathrm{H}\bx\right\Vert^2
		\notag \\&
		>0, \label{lemma 1}
	\end{align}
	where the inequality in (\ref{lemma 1}) implies $\bA+\bB\bB^\mathrm{H}$ is also a positive definite matrix. Since a positive definite matrix is invertible, $\bA+\bB\bB^\mathrm{H}$ is an invertible matrix, which finishes the proof.
\end{proof}

In (\ref{after Sylvester's det thm}), since $\det \left(\diag\left([\kappa_\ell,\kappa_\ell^*]^\mathrm{T}\right)\right)$$=\vert e^{j\phi_{\ell}}\lambda\vert^2$ and $\det(\bA_\ell)$ are constants and independent of $\phi_\ell$, the optimization problem in (\ref{IRS phase optimization problem}) can be represented as
\begin{align}
	\max_{\phi_\ell}\enspace \det \left(
	\begin{bmatrix}
		\frac{e^{-j\phi_{\ell}}}{\lambda} & 0 \\
		0 & \frac{e^{j\phi_{\ell}}}{\lambda}
	\end{bmatrix}
	+\left[\bq_\ell,\bp_\ell\right]^\mathrm{H} \bA_\ell^{-1} \left[\bp_\ell, \bq_\ell \right]\right),
\end{align} 
and the optimal phase shift $\phi_\ell^\star$ can be obtained as
\begin{align}
	&\phi_\ell^\star \notag \\
	=&\argmax_{\phi_\ell} \det \left(
	\begin{bmatrix}
		\frac{e^{-j\phi_{\ell}}}{\lambda}+\bq_\ell^\mathrm{H} \bA_\ell^{-1}\bp_\ell & \bq_\ell^\mathrm{H} \bA_\ell^{-1}\bq_\ell \\
		\bp_\ell^\mathrm{H} \bA_\ell^{-1}\bp_\ell & \frac{e^{j\phi_{\ell}}}{\lambda}+\bp_\ell^\mathrm{H} \bA_\ell^{-1}\bq_\ell
	\end{bmatrix}
	\right)\notag\\
	\stackrel{(\mathrm{a})}{=}&\argmax_{\phi_\ell} \mathrm{Re}\left(\frac{e^{-j\phi_{\ell}}}{\lambda}\bp_\ell^\mathrm{H}\bA_\ell^{-1}\bq_\ell\right) \notag \\
	=&\angle \left( \bp_\ell^\mathrm{H} \bA_\ell^{-1} \bq_\ell \right) \notag \\
	=&\angle \Bigg( \bh_{\mathrm{IB},\ell}^\mathrm{H} \bigg\{ \bI_N+\lambda \bigg(\bH_{-\ell}\bH_{-\ell}^\mathrm{H}+\bh_{\mathrm{IB},\ell}\bh_{\mathrm{UI},\ell}^\mathrm{H}
	\notag \\	& \times
	\Big(\bh_{\mathrm{IB},\ell}\bh_{\mathrm{UI},\ell}^\mathrm{H}\Big)^\mathrm{H}\bigg)\bigg\}^{-1} \bH_{-\ell}\bh_{\mathrm{UI},\ell}\Bigg), \label{optimal phi}
\end{align}
where (a) can be derived by straightforward linear operations. Although the optimal value $\phi_\ell^\star$ can be derived by~(\ref{optimal phi}), the solution requires the BS to know $\bh_{\mathrm{IB},\ell}$ and $\bh_{\mathrm{UI},\ell}$ to compute~$\phi_\ell^\star$. When the BS has the channel information in the form of the rank-one matrices $\bR_\ell$ instead of $\bH_{\mathrm{IB}}$ and $\bH_{\mathrm{UI}}$, the BS is able to get the optimal~$\phi_{\ell}^\star$~as
\begingroup
\begin{align}
	\phi_\ell^\star 
	=&\angle \Bigg(\mathrm{Tr}\Bigg( \bh_{\mathrm{IB},\ell}^\mathrm{H} \bigg\{ \bI_N+\lambda \bigg(\bH_{-\ell}\bH_{-\ell}^\mathrm{H}+\bh_{\mathrm{IB},\ell}\bh_{\mathrm{UI},\ell}^\mathrm{H}		
	\notag \\	&	\times
	\Big(\bh_{\mathrm{IB},\ell}\bh_{\mathrm{UI},\ell}^\mathrm{H}\Big)^\mathrm{H}\bigg)\bigg\}^{-1} \bH_{-\ell}\bh_{\mathrm{UI},\ell} \Bigg)\Bigg)\notag \\
	\stackrel{(\mathrm{a})}{=}&\angle \Bigg(\mathrm{Tr}\Bigg( \bh_{\mathrm{UI},\ell}\bh_{\mathrm{IB},\ell}^\mathrm{H} \bigg\{ \bI_N+\lambda \bigg(\bH_{-\ell}\bH_{-\ell}^\mathrm{H}+\bh_{\mathrm{IB},\ell}\bh_{\mathrm{UI},\ell}^\mathrm{H}
	\notag \\	& \times
	\Big(\bh_{\mathrm{IB},\ell}\bh_{\mathrm{UI},\ell}^\mathrm{H}\Big)^\mathrm{H}\bigg)\bigg\}^{-1} \bH_{-\ell} \Bigg)\Bigg) \notag \\
	=&\angle \left( \mathrm{Tr}\left( \bR_{\ell}^\mathrm{H}\Big(\bI_N+\lambda \bH_{-\ell}\bH_{-\ell}^\mathrm{H}+\lambda\bR_{\ell}\bR_{\ell}^\mathrm{H}\Big)^{-1} \right)\right), \label{optimal phi final}
\end{align}
where the property $\mathrm{Tr}(\bA\bB)=\mathrm{Tr}(\bB\bA)$ for any matrix $\bA$ and~$\bB$ whose multiplication produces a square matrix is used in (a). 
\endgroup


For given $\{\phi_k\}_{k=1, k\neq \ell}^{L}$, the optimal phase shift of the $\ell$-th IRS element $\phi_\ell^\star$ is given by (\ref{optimal phi final}). Then, we can derive the optimal values of all $L$ phase shifts in an iterative way. The proposed IRS phase shift design algorithm is summarized in Algorithm~1. Note that the optimality in (\ref{optimal phi final}) ensures that every update of $\phi_{\ell}^\star$ in Algorithm 1 improves the spectral efficiency until the algorithm stops. The algorithm can stop when the outer iteration index $i$ reaches its maximum value $I$ or when the sum of differences between the previous and updated phase shifts becomes less than a positive number $\epsilon$. With the designed phase shifts $\phi_{\ell}^\star$, the transmit beamformer is obtained as $\bW$ in (\ref{transmit beamformer}).


\alglanguage{pseudocode}
\begin{algorithm}[t]
	\caption{Proposed phase shift design at the IRS}
	\textbf{Initialize}
	\begin{algorithmic}[1]
		\State Set $\epsilon >0$
		\For {$\ell=1,\cdots,L$}
		\State $\bH_{-\ell}=\begin{cases}
			\bH_\mathrm{UB}, & \ell=1 \\
			\bH_\mathrm{UB}+\sum_{k=1}^{\ell-1}e^{j\phi_k^\star}\bR_k, & \text{else}
		\end{cases}$
		\State Update $\phi_\ell^\star$ by (\ref{optimal phi final})
		\State Set $\phi_{\mathrm{tmp},\ell}=\phi_\ell^\star$
		\EndFor
	\end{algorithmic}
	\textbf{Iterative update}
	\begin{algorithmic}[1]
		\addtocounter{ALG@line}{+6}
		\For{$i=1,\cdots,I$}
		\For{$\ell=1,\cdots,L$}
		\State $\bH_{-\ell}=\bH_\mathrm{UB}+\sum_{k=1,k\neq \ell}^{L}e^{j \phi_k^\star} \bR_k$
		\State Update $\phi_\ell^\star$ by (\ref{optimal phi final})
		\EndFor
		\If {{$\sum_{\ell=1}^{L}\left\vert\phi_\ell^\star-\phi_{\mathrm{tmp},\ell}\right\vert < \epsilon$}}
		\State Break
		\Else
		\State Set $\phi_{\mathrm{tmp},\ell}=\phi_\ell^\star$
		\EndIf
		\EndFor
		\State Return $\phi_\ell^\star \enspace \forall \ell \in \{1,\cdots,L\}$
	\end{algorithmic}
\end{algorithm}
\begingroup\allowdisplaybreaks
With regard to the uplink data transmission, the phase shifts $\phi_{\ell}$ and uplink transmit beamformer $\bF$ can be similarly designed. The uplink spectral efficiency $R_\mathrm{UL}$ is given by
\begin{align}
	R_\mathrm{UL}= \log_2 \left(\det\left(\bI_r+\frac{P_\mathrm{UL}}{r N_0}\bF^\mathrm{H}\bH_\mathrm{tot}^\mathrm{H} \bH_\mathrm{tot} \bF \right)\right) \label{uplink spectral efficiency},
\end{align} 
where $\bF\in$~$\mathbb{C}^{M \times r}$ is given by $\bF=\bU_\mathrm{tot}(:,1:r)$, and $\bU_\mathrm{tot}$ is given in (\ref{svd}). Exploiting the same property used in (\ref{spectral efficiecny w/o beamformer}), we can represent $R_\mathrm{UL}$~as
\begin{align}
	R_\mathrm{UL}
	=&\log_2 \left(\det\left(\bI_M+\frac{P_\mathrm{UL}}{r N_0}\bU_\mathrm{tot}^\mathrm{H}\bH_\mathrm{tot}^\mathrm{H} \bH_\mathrm{tot} \bU_\mathrm{tot}\right)\right) \notag\\
	=&\log_2 \left(\det\left(\bI_M+\frac{P_\mathrm{UL}}{r N_0}\bH_\mathrm{tot}^\mathrm{H} \bH_\mathrm{tot} \right)\right) \notag \\
	\stackrel{(\mathrm{a})}{=}&\log_2 \left(\det\left(\bI_N+\frac{P_\mathrm{UL}}{r N_0}\bH_\mathrm{tot} \bH_\mathrm{tot}^\mathrm{H} \right)\right)  \notag\\
	=&\log_2\left( \det\left(\bI_N + \frac{P_\mathrm{UL}}{r N_0}\bV_\mathrm{tot}^\mathrm{H}\bH_\mathrm{tot}\bH_\mathrm{tot}^\mathrm{H}\bV_\mathrm{tot}\right)\right), \label{uplink spectral efficiency final}
\end{align}         
where (a) can be derived using the Sylvester's determinant theorem \cite{matrix_book_2}. For the same transmit power $P_\mathrm{UL}=P_\mathrm{DL}$,
the uplink spectral efficiency in (\ref{uplink spectral efficiency final}) becomes the same as the downlink spectral efficiency in (\ref{spectral efficiecny w/o beamformer}). This implies that the optimal phase shift in (\ref{optimal phi final}) also maximizes $R_\mathrm{UL}$, and the BS can use the same IRS phase shifts for both the uplink and downlink data transmissions. 

\endgroup


\section{Numerical Results} \label{numerical results}

In this section, we investigate the proposed IRS phase shift design and compare the channel estimation performance of proposed SPAC and SEROM with those of existing estimation techniques.
Regarding the UPA structure, we consider $N=N_\mathrm{v}\times N_\mathrm{h}$ antennas for the BS, $M=M_\mathrm{v}\times M_\mathrm{h}$ antennas for the UE, and $L=L_\mathrm{v}\times L_\mathrm{h}$ elements for the IRS. \red{For the uplink and downlink training sequences, we exploit DFT matrices with proper sizes depending on channel estimation techniques.} The $B$-bit quantization of each IRS phase shift is realized by rounding off to the nearest quantized value in $\left\{ 0, \frac{2\pi}{2^B}, \cdots, \frac{(2^{B}-1)2\pi}{2^B}\right\}$. The Rician fading channel is established with $K_\mathrm{UI}=5$ dB, $K_\mathrm{IB} = 5$ dB, and $K_\mathrm{UB} = 3$ dB where $d_\mathrm{UI}$, $d_\mathrm{IB}$, and $d_\mathrm{UB}$ are uniformly distributed in $[5, 10]$, $[90, 100]$, and $[d_\mathrm{IB}-d_\mathrm{UI},d_\mathrm{IB}+d_\mathrm{UI}]$ in the meter scale. For each channel, the number of NLoS paths is set as $G_\mathrm{UI}=4$, $G_\mathrm{IB}=4$, and $G_\mathrm{UB}=7$. The path-loss exponents for the large scale fading are set as $\eta_\mathrm{UI} = 2.2$, $\eta_\mathrm{IB} = 2.5$, and $\eta_\mathrm{UB} = 4.5$, and the path-loss is $\mu_0=-30$ dB at the unit distance $d_0=1$ m. The noise variance is $N_0=-89$~dBm.

The training sequence length of direct channel estimation is set as $\tau_{\mathrm{d}}=M$, and those of the rank-one channel estimations are set as $\tau_{\text{OBO}}=LM$, $\tau_{\text{Co-OBO}}=2L$, $\tau_{\mathrm{SPAC}}=(L_\mathrm{v}+L_\mathrm{h}-1)M$, and $\tau_{\mathrm{SEROM}}=QM$. With a configurable training sequence length, that of SEROM is simply set as $\tau_{\mathrm{SEROM}}=\tau_{\mathrm{SPAC}}$ by setting the parameter $Q=L_\mathrm{v}+L_\mathrm{h}-1$.

\beginbreak
\subsection{Investigation of the proposed IRS phase shift design}

We evaluate the spectral efficiency of proposed IRS phase shift design in Section~\ref{IRS Beamforming Techniques} and compare the result with that of the algorithm in \cite{IRS_overview_4}. The design purpose of algorithm in \cite{IRS_overview_4} is the maximization of spectral efficiency where the transmit beamformer and the IRS phase shifts are alternately updated until convergence. To solely compare the IRS element design performance, we operate the proposed IRS phase shift design and the algorithm in \cite{IRS_overview_4} with the perfect channel information, i.e., $\widehat{\bH}_\mathrm{UB}={\bH}_\mathrm{UB}$ and $\widehat{\bR}_\ell=\bR_\ell$. Without channel estimation, the spectral efficiency is computed as \eqref{spectral efficiency} in Section~\ref{IRS Beamforming Techniques}.
\endgroup
In Fig.~\ref{Fig. Full vs. exhaustive}, two-bit phase quantization $B=2$ is considered, and the maximum spectral efficiency, which is found by exhaustive search of all the possible quantized IRS phase shifts, is demonstrated for the reference. The spectral efficiencies of proposed phase shift design and algorithm in \cite{IRS_overview_4} are very close to the result of exhaustive search. 
This means that the two techniques provide proper IRS phase shifts to maximize the spectral efficiency.

\begin{figure}[t!]
	\centering
	\includegraphics[width=.45\textwidth]{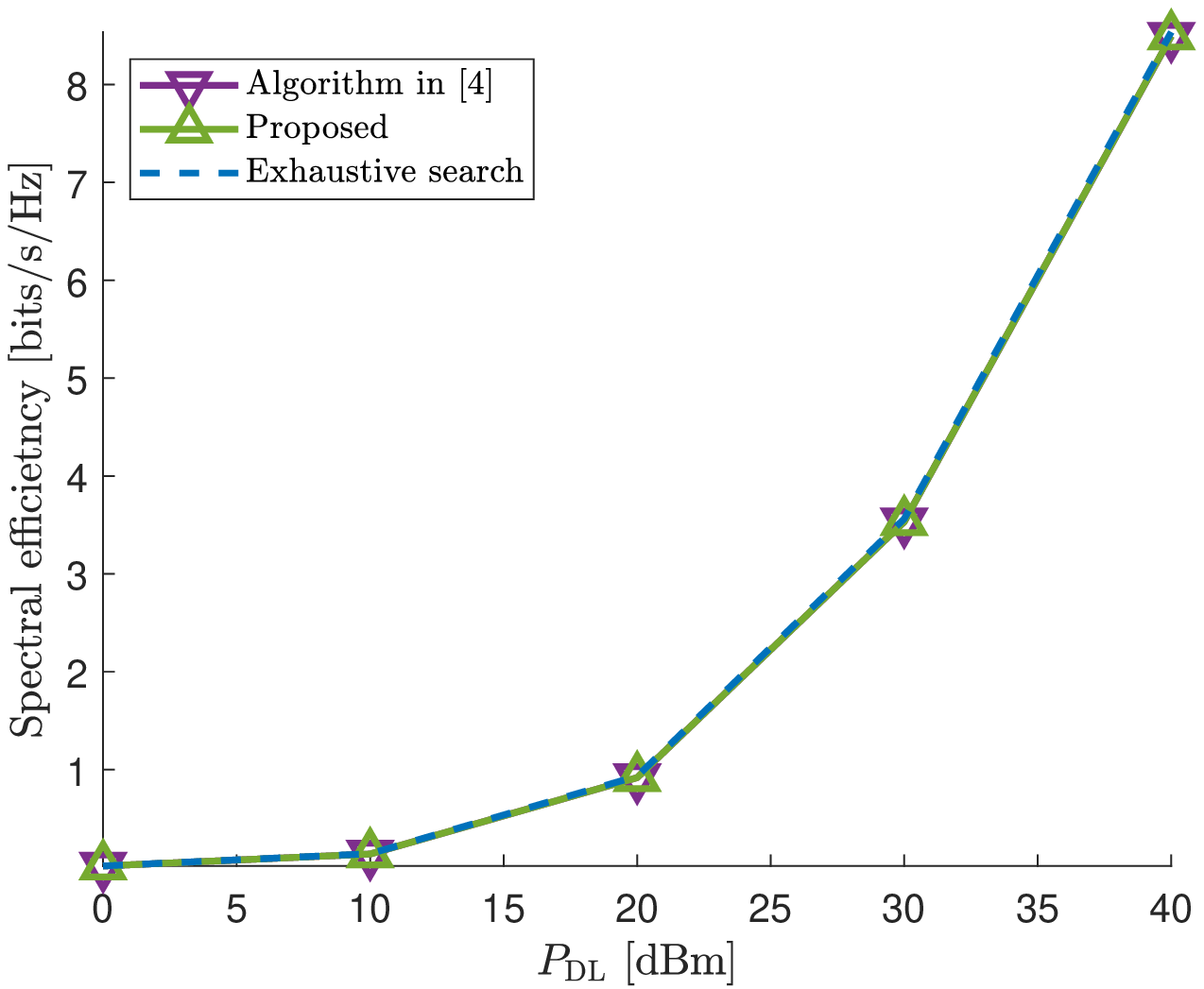} 
	\caption{Spectral efficiency with full channel information with $N_\mathrm{v}\times N_\mathrm{h}=2\times2$, $M_\mathrm{v}\times M_\mathrm{h}=1\times2$, $L_\mathrm{v}\times L_\mathrm{h}=3\times3$, and $B=2$.}
	\label{Fig. Full vs. exhaustive}
\end{figure}

\begin{table}[t]
	\renewcommand{\arraystretch}{2} 
	\centering
	\caption{Computation complexity of IRS phase shift techniques with $M<N<L$.}
	\begin{tabular}{| c | c |}
		\hline
		\multirow{2}*{\shortstack{IRS phase\\ shift techniques}} & \multirow{2}*{Number of scalar multiplications} \\ &\\
		\hline
		\hline
		\multirow{1}*{\shortstack{Algorithm in \cite{IRS_overview_4}}} & \multirow{1}*{$\mathcal{P}\left(I_\mathrm{init}LMN+I_\mathrm{outer}L(4M^2N+3M^3)\right)$}\\
		\hline
		\multirow{2}*{\shortstack{Proposed phase\\ shift design}} & 
		\multirow{2}*{ $\mathcal{P}\left(IL(3M^2N+2M^3)\right)$} \\& \\
		\hline
		\multirow{1}*{\shortstack{Exhaustive  search}}&
		\multirow{1}*{ $\mathcal{P}\left(2^{LB}(LMN)\right)$} \\ 
		\hline
	\end{tabular}
	\label{Table. complexity}
\end{table}

Note that Fig.~\ref{Fig. Full vs. exhaustive} only considers a small number of IRS elements due to the complexity of exhaustive search. While the proposed design and the algorithm in \cite{IRS_overview_4} both can serve a large number of IRS elements, there is difference on the computation complexity.
In Table~\ref{Table. complexity}, the computation complexity of three IRS phase shift techniques is listed by counting the number of scalar \red{multiplications, i.e., the} notation $\mathcal{P}(x)$ means that the number of scalar multiplications is proportional to $x$. The complexity of exhaustive search is remarkably higher than the other two techniques; it increases exponentially with the phase quantization bits and the number of IRS elements. The computation complexity of algorithm in \cite{IRS_overview_4} contains the parameters $I_\mathrm{init}$ and $I_\mathrm{outer}$ that are the numbers of initial random generations and outer algorithm iterations. 
The random initialization of algorithm in~\cite{IRS_overview_4} is to find good initial values that can reduce the number of outer algorithm iterations. On account of the interdependency of IRS phase shifts and transmit beamformer in the algorithm in~\cite{IRS_overview_4}, searching good initial values 
is not easy and results in additional complexity. The proposed IRS phase shift design can be designed independently from the specific value of transmit beamformer, and its simple initialization contributes to the low complexity. Consequently, for the same algorithm iteration $I=I_\mathrm{outer}$, the complexity of proposed phase shift design is lower than that of the algorithm in~\cite{IRS_overview_4} where the difference is $\mathcal{P}(I_\mathrm{init}LMN+I_\mathrm{outer}L(M^2N+M^3))$. The difference grows with the number of antennas, IRS elements, and algorithm iterations, and it becomes significant when a large number of antennas and IRS elements are deployed. Therefore, we use the proposed phase shift design, which gives the similar result to the exhaustive search but operates with the lowest complexity among the three, to compare the channel estimation techniques.


\subsection{Comparison of channel estimation techniques}

\beginbreak
In this subsection, we compare the channel estimation performance of proposed SPAC and SEROM with those of existing estimation techniques in~\cite{G.T.deAraujo:2020-LSKRF,Q.Nadeem:2020-MMSE-DFT,Z.Wang:2020-low_overhead_estimation}. The results of elementary techniques in Sections~\ref{Tech1} and \ref{Tech2} are also depicted as references. In \cite{G.T.deAraujo:2020-LSKRF}, the least squares Khatri-Rao factorization (LSKRF) is proposed to estimate $\bH_{\mathrm{IB}}$ and $\bH_{\mathrm{UI}}$. In \cite{Q.Nadeem:2020-MMSE-DFT}, $\bH_{\mathrm{IB}}$ is assumed as a known LoS channel, and $\bH_{\mathrm{UB}}$ and $\bH_{\mathrm{UI}}$ are assumed as Rayleigh fading channels. Based on these assumptions, the MMSE-DFT is proposed to estimate $\bH_{\mathrm{UB}}$ and $\bH_{\mathrm{UI}}$. Without considering any training overhead, the two estimation techniques in \cite{G.T.deAraujo:2020-LSKRF} and \cite{Q.Nadeem:2020-MMSE-DFT} require the training sequence lengths $\tau_{\mathrm{c}}=\tau_\mathrm{LSKRF}=\tau_{\mathrm{MMSE}\text{-}\mathrm{DFT}}=LM$ that are clearly longer than those of SPAC $\tau_{\mathrm{SPAC}}=(L_\mathrm{v}+L_\mathrm{h}-1)M$ and SEROM $\tau_{\mathrm{SEROM}}=QM$. In \cite{Z.Wang:2020-low_overhead_estimation}, the three-phase channel estimation is designed with the relatively short training sequence length $\tau_{\mathrm{three}\text{-}\mathrm{phase}}=M+L+\max\left\{M-1,\left\lceil{\frac{(M-1)L}{N}}\right\rceil\right\}$, which is comparable to those of SPAC and SEROM depending on the number of antennas and IRS elements. As a baseline, the result of all-zero IRS phase setting $\phi_\ell=0$ for all $\ell$ is provided where the channel estimation is conducted only for the resulting total channel $\bH_{\mathrm{UB}}+\bH_{\mathrm{IB}}\bI_L\bH_{\mathrm{UI}}$ with the training sequence length $\tau_{\text{all-zero}}=M$. 

To analyze the performance of channel estimation techniques, we adopt three performance metric: spectral efficiency per channel use, training sequence length, and effective spectral efficiency. The first metric measures the effectiveness of estimated channels to design IRS phase shifts, and the second metric assesses the training overhead of estimation technique. The third metric jointly evaluates the estimated channels and the training overhead of estimation techniques.

\begin{figure}[t]
	\centering 
	\subfloat[$N_\mathrm{v}\times N_\mathrm{h}=2\times4$, $M_\mathrm{v}\times M_\mathrm{h}=2\times2$, $L_\mathrm{v}\times L_\mathrm{h}=2\times4$, and $B=2$]{
		\includegraphics[width=.45\textwidth]{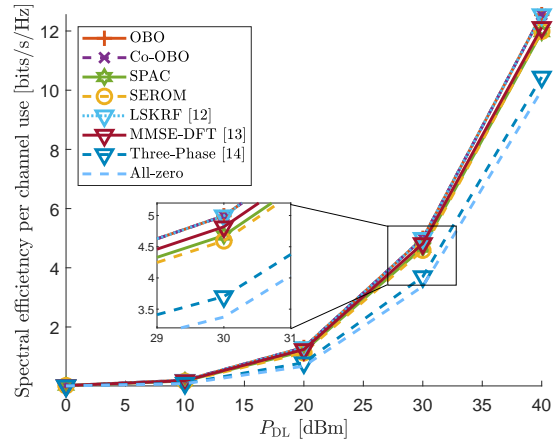}
	}
	\\
	\subfloat[$N_\mathrm{v}\times N_\mathrm{h}=4\times8$, $M_\mathrm{v}\times M_\mathrm{h}=4\times4$, $L_\mathrm{v}\times L_\mathrm{h}=8\times16$, and $B=4$]{
		\includegraphics[width=.45\textwidth]{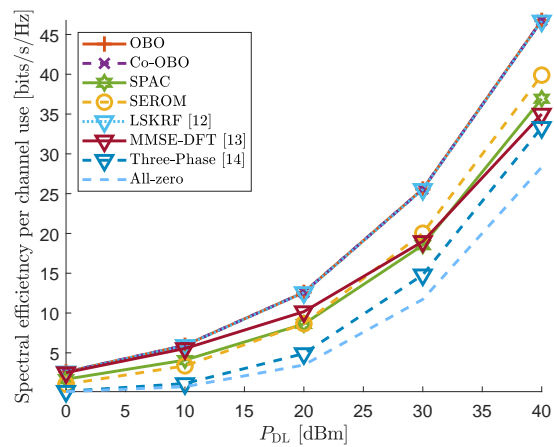}
	}
	\caption{Comparison of spectral efficiencies per channel \red{use}.}
	\label{Fig. spectral efficiency per channel use}
\end{figure}

\subsubsection{Spectral efficiency per channel use}
Based on estimated channels $\widehat{\bH}_{\mathrm{UB}}$ and $\widehat{\bR}_{\ell}$ for $\ell\in\{1,\cdots,L\}$, the spectral efficiency per channel use can be computed as~\eqref{spectral efficiency} by replacing $r$ and ${\bW}$ with $\hat{r}$ and $\widehat{\bW}$ where $\hat{r}$ is the rank of estimated channel $(\widehat{\bH}_{\mathrm{UB}}+\sum_{\ell=1}^L e^{j\phi_\ell}\widehat{\bR}_\ell)^\mathrm{H}$ and $\widehat{\bW}$ is composed of the first $\hat{r}$ right singular vectors of $(\widehat{\bH}_{\mathrm{UB}}+\sum_{\ell=1}^L e^{j\phi_\ell}\widehat{\bR}_\ell)^\mathrm{H}$ corresponding to the $\hat{r}$ dominant singular values. 
In Fig.~\ref{Fig. spectral efficiency per channel use}, the spectral efficiencies per channel use are depicted for two cases: one with a small number of antennas and IRS elements and the other with a large number of antennas and IRS elements. The Co-OBO estimation shows the highest spectral efficiency, but this is due to the ideal analog feedback, which is difficult to achieve in practice. With high training overhead, the OBO estimation and the LSKRF in~\cite{G.T.deAraujo:2020-LSKRF} also provide high spectral efficiencies. The MMSE-DFT in~\cite{Q.Nadeem:2020-MMSE-DFT} is another technique that requires high training overhead, but its spectral efficiency is lower than those of the OBO estimation and the LSKRF. This is because the MMSE-DFT is based on the assumption of Rayleigh fading, which deteriorate the estimation accuracy for the Rician fading with the LoS path. With low training overhead, SPAC and SEROM give moderate spectral efficiencies similar to the MMSE-DFT. The spectral efficiency of three-phase channel estimation in~\cite{Z.Wang:2020-low_overhead_estimation} is lower than other channel estimation techniques. The three-phase estimation requires arbitrary $M$ columns of $\bH_{\mathrm{IB}}$ to be linearly independent, but this condition is rarely satisfied without rich scattering environments. With ineffective IRS phase shifts, the all-zero phase setting provides the lowest spectral efficiency per channel use.

\pgfplotsset{
	/pgfplots/area legend/.style={
		/pgfplots/legend image code/.code={
			\fill[##1] (0cm,0.3em) rectangle (0.7em,-0.1em);
		}, },
}
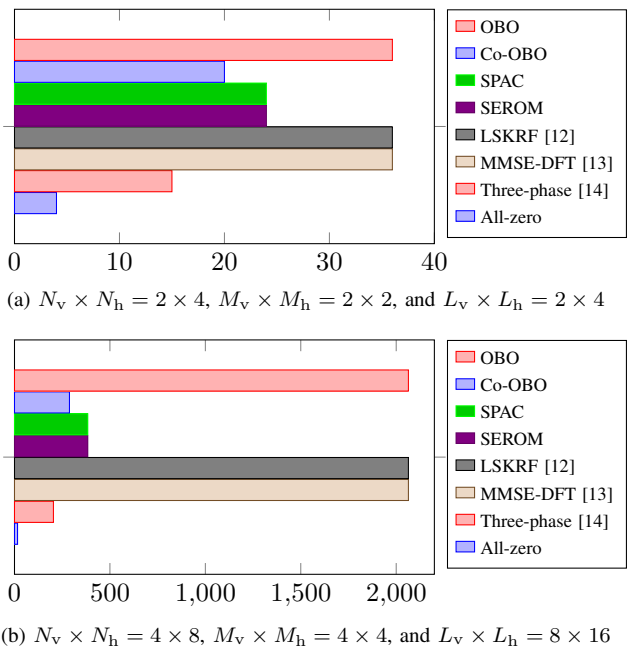
\begin{figure}[t]
	\centering
	\subfloat[$N_\mathrm{v}\times N_\mathrm{h}=2\times4$, $M_\mathrm{v}\times M_\mathrm{h}=2\times2$, and $L_\mathrm{v}\times L_\mathrm{h}=2\times4$]{
		\begin{tikzpicture}
			\begin{axis} [
				reverse legend,
				legend image code/.code={
					\draw [#11] (0cm,0.1cm) rectangle (0.0cm,0.1cm);
				},
				legend style = {font=\fontsize{7}{5}\selectfont},
				legend pos = outer north east,
				legend cell align={left},
				width = .395\textwidth,
				height = 4.7cm,
				xbar = .01cm,
				bar width = 8pt,
				xmin = 0,
				xmax = 40,
				ytick = data,
				symbolic y coords = {},
				area legend,		
				]
				\addplot coordinates {(4,)};
				\addplot coordinates {(15,)};
				\addplot coordinates {(36,)};
				\addplot coordinates {(36,)};
				\addplot coordinates {(24,)};
				\addplot coordinates {(24,)};
				\addplot coordinates {(20,)};
				\addplot coordinates {(36,)};
				
				\legend {All-zero,Three-phase\cite{Z.Wang:2020-low_overhead_estimation},MMSE-DFT\cite{Q.Nadeem:2020-MMSE-DFT},LSKRF\cite{G.T.deAraujo:2020-LSKRF},SEROM,SPAC,Co-OBO,OBO};
			\end{axis}
		\end{tikzpicture}
	}
	\\
	\subfloat[$N_\mathrm{v}\times N_\mathrm{h}=4\times8$, $M_\mathrm{v}\times M_\mathrm{h}=4\times4$, and $L_\mathrm{v}\times L_\mathrm{h}=8\times16$]{
		\begin{tikzpicture}
			\begin{axis} [
				reverse legend,
				legend style = {font=\fontsize{7}{5}\selectfont},
				legend pos = outer north east,
				legend cell align={left},
				width = .395\textwidth,
				height = 4.7cm,
				xbar = .01cm,
				bar width = 8pt,
				xmin = 0,
				xmax = 2200,
				ytick = data,
				symbolic y coords = {},
				area legend,		
				]
				\addplot coordinates {(16,)};
				\addplot coordinates {(204,)};
				\addplot coordinates {(2064,)};
				\addplot coordinates {(2064,)};
				\addplot coordinates {(384,)};
				\addplot coordinates {(384,)};
				\addplot coordinates {(288,)};
				\addplot coordinates {(2064,)};
				
				\legend {All-zero,Three-phase\cite{Z.Wang:2020-low_overhead_estimation},MMSE-DFT\cite{Q.Nadeem:2020-MMSE-DFT},LSKRF\cite{G.T.deAraujo:2020-LSKRF},SEROM,SPAC,Co-OBO,OBO};
			\end{axis}
		\end{tikzpicture}
	}
	\caption{Comparison of training sequence lengths $\tau_{\mathrm{tot}}=\tau_{\mathrm{d}}+\tau_{\mathrm{c}}$.}
	\label{Fig. training sequence length}
\end{figure}

\subsubsection{Training sequence length}
For the two cases in Fig.~\ref{Fig. spectral efficiency per channel use}, the training sequence lengths of channel estimation techniques are compared in Fig.~\ref{Fig. training sequence length}. The training sequence length is computed as $\tau_{\mathrm{tot}}=\tau_d+\tau_c$ where $\tau_{\mathrm{d}}$ is to estimate the direct channel $\bH_{\mathrm{UB}}$ and $\tau_{\mathrm{c}}$ is to estimate the $L$ rank-one matrices $\bR_{\ell}$. The trend of training sequence lengths matches with that of spectral efficiencies per channel use in general. The OBO estimation and the LSKRF provide high spectral efficiencies per channel use, and their training overhead is higher than that of other techniques. The three-phase estimation and the all-zero phase setting have short training sequence lengths and provide lower spectral efficiencies per channel use than other techniques. The performance of SPAC and SEROM is in middle, but their spectral efficiencies per channel use are close to those of the OBO estimation and the LSKRF, and their training sequence lengths are close to those of the three-phase estimation and the all-zero phase setting.

In Fig.~\ref{Fig. training sequence length}, it is shown that the overhead of training grows with the number of antennas and IRS elements. However, the coherence time block length of typical communication system is hard to be longer than 1,200 or 2,400 \cite{S.Noh:2020-Block_length,3GPP_Block_length}, and the training sequence length longer than 2,400 would not be acceptable. For the second case with the large numbers of antennas and IRS elements, the training sequence lengths of OBO estimation, LSKRF, and MMSE-DFT are already over 2,000, which means the three estimation techniques have only a little time for data transmissions after channel estimation. 

\begin{figure}[t]
	\centering 
	\subfloat[$N_\mathrm{v}\times N_\mathrm{h}=2\times4$, $M_\mathrm{v}\times M_\mathrm{h}=2\times2$, $L_\mathrm{v}\times L_\mathrm{h}=2\times4$, $B=2$, and $\Gamma=150$]{
		\includegraphics[width=.45\textwidth]{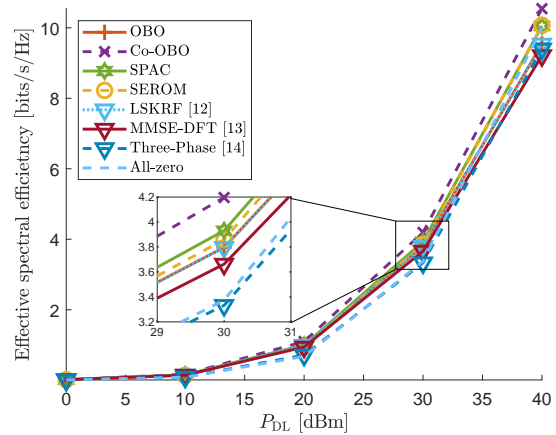}
		\label{Fig. spec. eff. small}
	}
	\\
	\subfloat[$N_\mathrm{v}\times N_\mathrm{h}=4\times8$, $M_\mathrm{v}\times M_\mathrm{h}=4\times4$, $L_\mathrm{v}\times L_\mathrm{h}=8\times16$, $B=4$, and $\Gamma=2,400$]{
		\includegraphics[width=.45\textwidth]{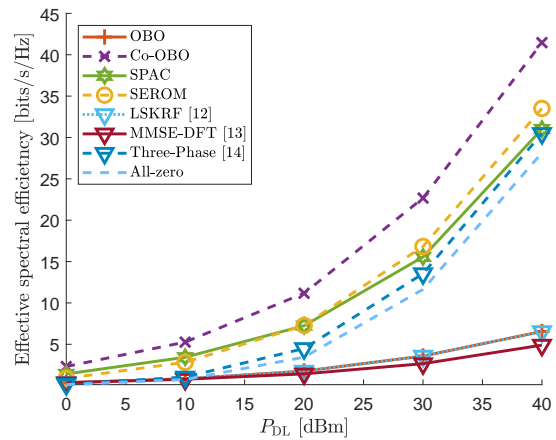}
		\label{Fig. spec. eff. large}
	}
	\caption{Comparison of effective spectral efficiencies.}
	\label{Fig. spec. eff.}
\end{figure}

\subsubsection{Effective spectral efficiency}

Now, we jointly assess the estimated channel and the training overhead by measuring the effective spectral efficiency as \cite{B.Hassibi:2003-training-base-spec-eff}
\begin{align}
	&\frac{\Gamma-\tau_{\mathrm{tot}}}{\Gamma}\log_2\Bigg(\det\Bigg(
	\bI_{\hat{r}} + \frac{P_{\mathrm{DL}}}{{\hat{r}} N_0}\widehat{\bW}^\mathrm{H}\left(\bH_{\mathrm{UB}}+\bH_{\mathrm{IB}}\boldsymbol{\Phi}\bH_{\mathrm{UI}}\right)
	\notag\\&\times
	\left(\bH_{\mathrm{UB}}+\bH_{\mathrm{IB}}\boldsymbol{\Phi}\bH_{\mathrm{UI}}\right)^\mathrm{H}\widehat{\bW}
	\Bigg)\Bigg) 
	\notag\\
	=&\frac{\Gamma-\tau_{\mathrm{tot}}}{\Gamma}\log_2\Bigg(\det\Bigg(
	\bI_{\hat{r}} + \frac{P_{\mathrm{DL}}}{{\hat{r}} N_0}\widehat{\bW}^\mathrm{H}\left(\bH_{\mathrm{UB}}+\sum_{\ell=1}^L e^{j\phi_\ell}\bR_{\ell}\right)
	\notag\\&\times
	\left(\bH_{\mathrm{UB}}+\sum_{\ell=1}^L e^{j\phi_\ell}\bR_{\ell}\right)^\mathrm{H}\widehat{\bW}
	\Bigg)\Bigg),\label{Eq. effective spectral efficiency}
\end{align}
\endgroup
where $\Gamma$ is the coherence time block length. In Fig.~\ref{Fig. spec. eff. small}, with its short training sequence length and high spectral efficiency per channel use, the Co-OBO estimation shows the highest spectral efficiency, which results from the ideal feedback. SPAC and SEROM also provide high spectral efficiencies with low training overhead. With the small number of IRS elements and low bits for quantization $B=2$, the IRS training matrix $\boldsymbol{\Omega}$ for SEROM in \eqref{Omega} is hard to satisfy the pseudo-orthogonality condition $|a_\ell|\gg|b_{\ell,k}|$, and this can degrade the estimation accuracy of SEROM. On the contrary, SPAC is not influenced by the phase quantization at all in channel estimation, and it gives a little higher spectral efficiency than SEROM. Since the OBO estimation, the LSKRF, and the MMSE-DFT consume high training overhead, their effective spectral efficiencies become lower than those of SPAC and SEROM. The all-zero phase setting and the three-phase estimation operate with short training sequence lengths, but this advantage barely compensates for their ineffective IRS phase shifts.

In Fig.~\ref{Fig. spec. eff. large}, the spectral efficiencies with large numbers of antennas and IRS elements are depicted. The Co-OBO estimation, SPAC, and SEROM still provide high spectral efficiencies as in Fig.~\ref{Fig. spec. eff. small}. With a large enough number of IRS elements and quantization bits, the IRS training matrix $\boldsymbol{\Omega}$ for SEROM can easily meet the pseudo-orthogonality condition $|a_\ell|\gg|b_{\ell,k}|$, and SEROM outperforms SPAC in this case. 
With little time for data transmissions, the spectral efficiencies of OBO estimation, LSKRF, and MMSE-DFT are significantly reduced. By the same token, the spectral efficiency of three-phase estimation is relatively improved with its short training sequence length, compensating for poor channel estimation performance. The spectral efficiency of MMSE-DFT falls below all the other techniques as transmit power increases. This is because the large channel dimension deepens the gap between the Rician channel structure and the supposed channel structure of MMSE-DFT. At high transmit power, even the all-zero phase setting provides higher spectral efficiency than the techniques requiring high training overhead.

\begin{figure}[t]
	\centering 
	\includegraphics[width=.45\textwidth]{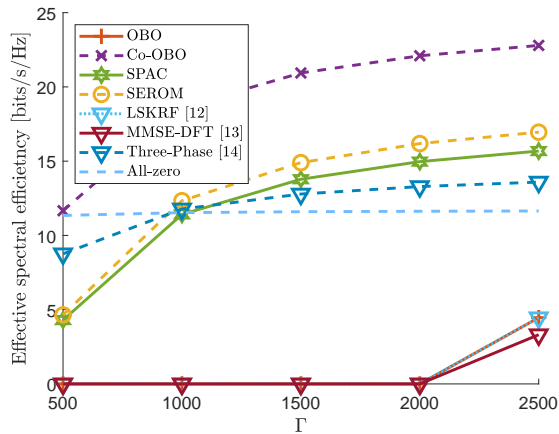} 
	\caption{Effective spectral efficiency vs. $\Gamma$ at $P_\mathrm{DL}=30$ dBm with $N_\mathrm{v}\times N_\mathrm{h}=4\times8$, $M_\mathrm{v}\times M_\mathrm{h}=4\times4$, $L_\mathrm{v}\times L_\mathrm{h}=8\times16$, and $B=4$.}
	\label{Fig. spec. eff. over Gamma}
\end{figure}

In Fig.~\ref{Fig. spec. eff. over Gamma}, the effective spectral efficiencies of channel estimation techniques are compared over coherence time block length $\Gamma$. With very small $\Gamma$, the three-phase estimation and the all-zero phase setting, which have short training sequence lengths, provide high spectral efficiencies. As $\Gamma$ grows, spectral efficiencies of SPAC and SEROM increase with long time for data transmissions. The spectral efficiencies of OBO estimation, LSKRF, and MMSE-DFT are zero due to their long training sequence lengths until $\Gamma=2,000$. Except the Co-OBO that is impractical due to the ideal feedback, SPAC and SEROM provide the highest effective spectral efficiencies for most practical range of $\Gamma$. This is by virtue of a fine balance between the training sequence length and spectral efficiency per channel use that each of SPAC and SEROM provides. A better balance also can be found by adjusting the training sequence length of SEROM.
\begingroup\allowdisplaybreaks

\endgroup


\section{Conclusion} \label{conclusion}

We proposed two novel practical channel estimation techniques and an IRS phase shift design. The proposed SPAC and SEROM are designed to estimate channel information in SU-MIMO systems while consuming short training sequence lengths. The proposed IRS phase shift design is developed to maximize spectral efficiency while requiring only linear operations. Numerical results showed that the proposed phase shift design provides a spectral efficiency close to that of exhaustive search. When the proposed IRS phase shift design \red{was} utilized, the effective spectral efficiencies of SPAC and SEROM were higher than those of other estimation techniques. The results verified that the high spectral efficiency can be achieved by considering both the training overhead and the spectral efficiency per channel use. \red{A possible future work is to develop a joint} framework of channel estimation and IRS element design to have low training overhead while extracting only a necessary information to design IRS \red{elements, still achieving a high spectral efficiency.}
%

%
%


\bibliographystyle{IEEEtran}
\bibliography{refs_IRS_SU_MIMO_abb_etal2}

\end{document}